\documentclass[12pt]{article}
\usepackage{amsmath,amssymb}
\usepackage{graphicx}
\usepackage{amsthm}
\usepackage{ascmac}
\usepackage{braket}
\usepackage{authblk}
\usepackage{fullpage}
\usepackage{hyperref}
\usepackage{appendix}

\theoremstyle{definition}
\newtheorem{theorem}{Theorem}[section]
\newtheorem{definition}{Definition}[section]
\newtheorem{proposition}{Proposition}[section]
\newtheorem{example}{Example}[section]
\newtheorem{lemma}{Lemma}[section]
\newtheorem{remark}{Remark}[section]
\newtheorem{corollary}{Corollary}[section]

\newcommand{\Z}{\mathbb{Z}}
\newcommand{\C}{\mathbb{C}}

\numberwithin{equation}{section}

\begin{document}
\title{Degeneration limits of Virasoro vertex operators and Painlev\'e $\tau$ functions}
\author{Hajime Nagoya and Haruki Nakagawa \\ School of Mathematics and Physics, Kanazawa University, Kanazawa, Ishikawa 920-1192, Japan \\ E-mail: \href{mailto:nagoya@se.kanazawa-u.ac.jp}{nagoya@se.kanazawa-u.ac.jp}, \href{mailto:nakagawaharuki0@gmail.com}{nakagawaharuki0@gmail.com}}
\date{}
\maketitle
\begin{abstract}
    We construct degeneration limits of vertex operators for the Virasoro algebra. Our method relies on the rearranged expansion of compositions of vertex operators together with their integral representations. Using this framework, we obtain a vertex operator between Verma modules of rank $r+1$ as a degeneration of a composition of two vertex operators between Verma modules of rank $r$ ($r\in\Z_{\geq 0}$). Furthermore, we apply these degeneration limits to prove the conjectural expansions of the $\tau$ functions of the fifth and fourth Painlev\'e equations in terms of irregular conformal blocks \cite{Nagoya 2015}.
\end{abstract}

\tableofcontents
\section{Introduction}

\subsection{Virasoro vertex operators}

In two-dimensional conformal field theories, vertex operators are specific linear maps between Verma modules of the Virasoro algebra. We call the expectation values of such vertex operators conformal blocks, which are essential objects in two-dimensional conformal field theory, initiated by Belavin, Polyakov, and Zamolodchikov \cite{BPZ}, and play an indispensable role in mathematics and theoretical physics\footnote{
Unless otherwise stated, all variables and parameters are complex. Positivity or reality assumptions are imposed only when we choose explicit branches and contours for the free-field integral representations.}.

In 2009, Gaiotto recognized the need to consider conformal field theories with irregular singularities in connection with asymptotically free $N=2$ gauge theories \cite{Gaiotto}. To analyze irregular singularities, an irregular vector $\ket{\Lambda}$ is defined by the following conditions:
\begin{equation}\label{eq condition IV}
    L_n \ket{\Lambda}= \Lambda_n \ket{\Lambda} \quad (r\leq n\leq2 r), \quad L_n\ket{\Lambda}=0\quad (n>2r).
\end{equation}
Here, $L_n$ are the Virasoro generators.  The vector $\ket{\Lambda}$ is said to be of rank $r$ when $\Lambda_{2r} \neq 0$, and of rank $r - 1/2$ when $\Lambda_{2r} = 0$ and $\Lambda_{2r-1} \neq 0$.
 In 2012, Gaiotto and Teschner investigated irregular states and their collision limits in Liouville theory \cite{GT}.
Their rearranged expansion in the rank-one case suggested the degeneration mechanism that is used here for vertex-operator matrix elements. 

In 2015, a mathematically rigorous definition of an irregular vertex operator
\begin{align*}
    \Phi^{\Delta}_{\Lambda', \Lambda}(z): M^{[r]}_\Lambda \to M^{[r]}_{\Lambda'}
\end{align*}
was proposed in \cite{Nagoya 2015} for any positive integer $r$, where $\Delta\in\mathbb{C}$, and $M_{\Lambda}^{[r]}$ is an irregular Verma module of rank $r$ with weight $\Lambda=(\Lambda_{r},\Lambda_{r+1},\ldots,\Lambda_{2r})$. When $\Lambda_{2r}\ne0$, this paper also proved the existence and uniqueness of such operators. In the same work, irregular vertex operators that increase the rank of the irregular Verma module:
\begin{align*}
    \Phi^{[r],\lambda}_{\Lambda',\Delta }(z): M^{[0]}_\Delta \to M^{[r]}_{\Lambda'},
\end{align*}
were also defined. The existence and uniqueness of these irregular vertex operators of rank $r$ were proved in 2018  \cite{Nagoya 2018}.
Moreover, \cite{Nagoya 2018} introduced ramified irregular vertex operators of the Virasoro algebra
\begin{align*}
    \Phi^{\Delta}_{\Lambda, \Lambda'}(z): M^{[r]}_\Lambda \to M^{[r]}_{\Lambda'}\quad ( \Lambda=(\Lambda_{r},\Lambda_{r+1},\ldots,\Lambda_{2r-1})\in\C^{r-1}\times \C^*),
\end{align*}
which describes an irregular singularity of rank $r-1/2$.

In general, irregular conformal blocks are defined in terms of irregular vertex operators and irregular vectors.
 A mathematically rigorous framework for defining irregular conformal blocks was established in \cite{Nagoya 2015}, based on the exact definition and the unique existence of the irregular vertex operators given in Definition 2.10 and Theorem 2.12 of the same paper. Regarding the irregular vectors themselves, in the case of $r=1$, it is easy to see that the condition (1.1) uniquely determines the irregular vector $|\Lambda\rangle$ as an element of the completion of a Verma module. For higher-rank cases, obtaining their explicit expansions has been a subject of recent study. An ansatz expanding irregular vectors as a sum over generalized descendants was first proposed in \cite{GT}. This approach was subsequently generalized to an arbitrary integer rank $r$ in  \cite{Nishinaka Uetoko}, and has been further extended to half-integer ranks, such as $5/2$ and $r-1/2$, in \cite{PP 2023} and \cite{HNNT}, respectively. A recent algebraic construction proving
existence and uniqueness of integer- and half-integer-rank irregular vectors
was given in \cite{Nagoya-new}.

\subsection{Main results}

The first main result of this paper is the construction of degeneration limits for Virasoro vertex operators.
We rewrite the composition of two vertex operators in the form
\begin{equation}\label{eq rank 0 to 1 re ex intro}
  \Phi_{\Delta_5, \Delta_3}^{\Delta_4}(z)\,
  \Phi_{\Delta_3, \Delta_1}^{\Delta_2}(w)\ket{\Delta_1}
  =
  z^{\Delta_5 - \Delta_4 - \Delta_3}\,
  w^{\Delta_3 - \Delta_2 - \Delta_1}
  \left(1-\frac{w}{z}\right)^A
  \sum_{k=0}^{\infty} \ket{R_k(z)}\, w^{k},
\end{equation}
where $\ket{R_k(z)}$ are vectors in the Verma module $M_{\Delta_5}$. Here \(A\) is an auxiliary exponent chosen for the degeneration limit. Its role is to absorb the divergent part of the collision so that the coefficients \(\ket{R_k(z)}\) have finite limits. We refer to this as a rearranged expansion of the composition as in \cite{GT}. We consider the limit $z \to 0$ with suitable parameterization. In \cite{GT}, it was observed that the first terms of $\ket{R_k(z)}$ converge in the limit. Lisovyy, Nagoya, and Roussillon used this type of expansion to compute connection formulas for the $\tau$ function of the fifth Painlev\'e equation in \cite{Lisovyy Nagoya Roussillon}. In this rearranged expansion, the coefficients $\ket{R_k(z)}$ in \eqref{eq rank 0 to 1 re ex intro} satisfy the recursive relations. However, it turns out that even if we assume convergence of $\ket{R_0(z)}, \ldots, \ket{R_{k-1}(z)}$, it is not easy to see whether $\ket{R_k(z)}$ converges or not only from the recursive relations. Hence, recursive relations alone do not provide a convenient framework to prove the existence of degenerate limits of $\ket{R_k(z)}$. We need to calculate $\ket{R_k(z)}$ directly. For this reason, we use integral representations of vertex operators. From these integral representations, we justify the degenerate limits of the vertex operators from rank $0$ to rank $1$.

We first state the degeneration from a composition of regular vertex operators to an irregular vertex operator acting between rank-one modules.
\begin{theorem}\label{main result 1-1}
    Let the parameters be chosen as in Section \ref{subsec rvo degeneration}. Then, after applying the normalized rearranged expansion to the highest-weight vector
and expanding formally in \(w\), the coefficients \(\ket{R_k(z)}\) converge
coefficientwise to the coefficients of
\[
\Phi_{\Lambda',\Lambda}^{\Delta}(w)\ket{\Lambda}.
\]
\end{theorem}
The next theorem deals with the compositions of irregular vertex operators acting between irregular Verma modules of rank $r$. Here, the same difficulty appears as in the regular case: even if the lower coefficients in the rearranged expansion converge, the recursive relations do not directly imply the convergence of the next coefficient. We use the integral representation of the irregular vertex operators to overcome these problems.
\begin{theorem}\label{main result 1-2}
     Let the parameters be chosen as in Section \ref{subsec ivo degeneration}.
 Then, after applying the normalized rearranged expansion for $\Phi_{\Lambda',\tilde{\Lambda}} ^{\Delta_z}(z)\Phi_{\tilde{\Lambda},\Lambda} ^{\Delta_w}(w):M_{\Lambda}^{[r]} \to M_{\Lambda'}^{[r]}$ to the irregular vector
and expanding formally in \(w\), the coefficients \(\ket{R_k(z)}\) converge
coefficientwise to the coefficients of 
\[
\Phi_{\Gamma',\Gamma}^{\Delta_w}(w)\ket{\Gamma},
\]
where $\Phi_{\Gamma',\Gamma}^{\Delta_w}(w)$ is an irregular vertex operator from $M_\Gamma^{[r+1]}$ to $M_{\Gamma'}^{[r+1]}$. 
\end{theorem}
As a corollary of Theorems \ref{main result 1-1} and \ref{main result 1-2}, we obtain degenerations of conformal blocks. We may obtain an irregular vector by a certain limit of the action of an irregular vertex operator on an irregular vector by the method developed in this paper. We will report on this issue in the near future. It would also be desirable to understand whether ramified irregular vertex operators can also be realized by degeneration.
A direct application of the present method is obstructed by the fact that the ramified case involves expansions in half-integer powers, which requires further investigation.

\subsection{Painlev\'e equations and $\tau$ functions}

The second main result of this paper is an application of these degenerations of vertex operators to the Painlev\'e equations.

In \cite{GIL1}, Gamayun, Iorgov, and Lisovyy obtained expansions at $t=0$ of the fifth and third Painlev\'e equations $\mathrm{P_V},\mathrm{P_{{III}_{1,2,3}}}$ by taking a degenerate limit of this combinatorial formula of the $\tau$ function of the sixth Painlev\'e equation \cite{BS, GIL, ILT}.
The expansion is similar to the expansion of the sixth Painlev\'e equation
\begin{equation*}
    \tau_{\mathrm{J}} ^{(0)} (t)
    =\sum_{n\in\mathbb{Z}}s^n C_{\mathrm{J}}(\vec{\theta},\sigma+n)\mathcal{F}_{\mathrm{J}}^{(0)} (\vec{\theta},\sigma+n;t)
    \quad (\mathrm{J}=\mathrm{VI},\mathrm{V},\mathrm{III}_{1,2,3}), 
\end{equation*}
where the conformal block degenerates into the irregular conformal block. For example, in the case of the $\tau$ function of the fifth Painlev\'e equation, the following irregular conformal block appears:
\begin{align*}
    \mathcal{F}_{\mathrm{V}}^{(0)} (\theta,\sigma;t)
    &= \bra{(\theta_*, 1/4)} \cdot \left( \Phi_{\sigma^2 ,\theta_0^2}^{\theta_t^2}(t) \ket{\theta_0^2} \right).
\end{align*}

To describe the behavior of the $\tau$ functions at $t= \infty$, it is natural to consider other irregular conformal blocks that describe irregular singularities. One of the authors proposed a conjecture in 2015 for the expansion of the $\tau$ function of the fifth and fourth Painlev\'e equations at $t= \infty$ in terms of irregular conformal blocks  \cite{Nagoya 2015}. Moreover, by introducing ramified irregular vertex operators, series expansions of the $\tau$ functions of the third and second Painlev\'e equations $\mathrm{P_{III}}$ and $\mathrm{P_{II}}$ in terms of irregular conformal blocks of half rank type were also conjectured \cite{Nagoya 2018}.
Without using irregular vertex operators, one can construct irregular conformal blocks purely as pairings of irregular vectors within Virasoro Verma modules. Based on this approach, for the $\tau$ function of the first Painlev\'e equation, it was conjectured that it admits an expansion in terms of irregular conformal blocks constructed from a rank $5/2$ irregular vector \cite{PP 2023}.

We prove that the $\tau$ functions of the fifth and fourth Painlev\'e equations can be expressed as an expansion of irregular conformal blocks. We set the central charge $c=1$.
\begin{theorem}\label{main result 2-1}
A series expansion of the $\tau$ function of the fifth Painlev\'e equation at $t = \infty$ is given by
        \begin{align*}
            \tau_{\mathrm{V}}^{(\infty)} (t)
            =& \sum_{n \in \mathbb{Z}} e^{2\pi i n  \varrho} (-1)^{ \frac{1}{2}n(n+1)}   C_\mathrm{V}(\vec{\theta},\beta+n) \nonumber\\
            &\times \bra{ ( \theta, 1/4)} \Phi_{(  \theta, 1/4), \left( \theta - \beta -n,1/4 \right) } ^{\theta_t ^2,*} (t) \cdot\ket{\theta_0^2} ,
        \end{align*}
        where $\vec{\theta}$ stands for $(\theta, \theta_t, \theta_0)$, $\varrho,\, \beta \in \mathbb{C}$, and
        \begin{align*}
            C_\mathrm{V}(\vec{\theta},\beta)
            =& \prod_{\epsilon=\pm}G(1 + \epsilon \theta_0 + \theta - \beta )
            G(1+\theta_t +\epsilon \beta ).
        \end{align*}
\end{theorem}

\begin{theorem}\label{main result 2-2}
    A series expansion of the $\tau$ function of the fourth Painlev\'e equation at $t = \infty$ is given by
        \begin{align*}
            \tau_{\mathrm{IV}}^{(\infty)} (t)
            =& \sum_{n \in \mathbb{Z}} e^{2\pi i n  \varrho} C_\mathrm{IV}(\vec{\theta},\beta+n)  \braket{(\theta_*,0,1/4)| \Phi_{(\theta_*,0,1/4),(\theta_* - \beta - n,0,1/4)}^{\theta_t ^2,*}(t)\cdot |0},
        \end{align*}
        where  $\vec{\theta}$ stands for $(\theta_*,\theta_t)$, $\varrho,\, \beta \in \mathbb{C}$, and
        \begin{equation*}
            C_\mathrm{IV}(\vec{\theta},\beta)=G(1 + \theta_{*}  - \beta) \prod_{\epsilon=\pm} G(1+\theta_t +\epsilon \beta ).
        \end{equation*}
\end{theorem}

Here, the vector $\bra{(\theta, 1/4)}$ denotes the rank-one irregular vector defined by $(\Lambda_1, \Lambda_2)=(\theta, 1/4)$, and the vector $\bra{(\theta_*,0,1/4)}$ denotes the rank-two irregular vector defined by  $(\Lambda_2,\Lambda_3,\Lambda_4)=(\theta_*,0,1/4)$. Moreover, $G(x)$ is the Barnes $G$ function. The parameters $\theta, \theta_t,\theta_0,\theta_*$ are the parameters of the Painlev\'e equation, and $\varrho$, $\beta$ correspond to the initial conditions.

Theorems \ref{main result 2-1} and \ref{main result 2-2} agree with the conjectures formulated in \cite{Nagoya 2015}.
In particular, the above theorems establish that the $\tau$ functions of the fifth and fourth Painlev\'e equations admit a description at $t = \infty$ in terms of irregular conformal blocks.

This paper is organized as follows. In Section 2, we recall the representation theory of the Virasoro algebra and review regular and irregular vertex operators, along with their free field representations. In Section 3, we study the degenerate limits of vertex operators and prove that such limits exist under suitable conditions. In Section 4, we recall the Hamiltonian systems of the Painlev\'e equations, define the associated $\tau$ functions, and then apply the degenerations of vertex operators constructed in Section 3 to obtain the series expansions of the $\tau$ functions of the fifth and fourth Painlev\'e equations.

\par\medskip
\noindent\textbf{Acknowledgement.} 
 This work was supported by JSPS KAKENHI Grant Number 22K03350 and JST SPRING, Grant Number JPMJSP2135. 

\section{Vertex operators}
\subsection{Regular vertex operators}
The Virasoro algebra is the Lie algebra spanned by $L_n (n \in \mathbb{Z})$ and the central charge $C$, with commutation relations
\begin{align*}
 [L_m, L_n] &= (m-n)L_{m+n} +\frac{C}{12} (m^3-m)\delta_{m+n,0},\\
 [L_m,C] &= 0,
\end{align*}
where $\delta_{i,j}$ stands for Kronecker's delta.

Let $\mathrm{Vir}_{\geq 0}$ be the subalgebra generated by $L_n$ ($n\geq 0$).
A Verma module $M_{\Delta}$ with the highest weight $\Delta\in\C$ is the induced module
\begin{align*}
    M_{\Delta} &= \mathrm{Ind}_{\mathrm{Vir}_{\geq 0}} ^{\mathrm{Vir}} \mathbb{C} \ket{\Delta},
\end{align*}
where $L_n$ ($n\geq 0$) acts on the highest weight vector $\ket{\Delta}$ as
\begin{equation*}
    L_0 \ket{\Delta}= \Delta \ket{\Delta},\quad
    L_n \ket{\Delta}=0\quad (n\geq 1).
\end{equation*}

\begin{definition}\label{def reg vo}
A regular vertex operator $\Phi_{\Delta_3,\Delta_1}^{\Delta_2}(z) : M_{\Delta_1} \to M_{\Delta_3}$ is defined by
    \begin{align}
        [L_n,\Phi_{\Delta_3,\Delta_1}^{\Delta_2}(z)] &= z^n \left( z \dfrac{\partial}{\partial z} + (n+1) \Delta_2 \right) \Phi_{\Delta_3,\Delta_1}^{\Delta_2}(z), \label{eq reg vo}\\
        \Phi_{\Delta_3,\Delta_1}^{\Delta_2}(z) \ket{\Delta_1}
        &= z^{\Delta_3 - \Delta_2 -\Delta_1} \sum_{m \geq 0} \ket{v_m} z^m,\label{eq reg vo action}
    \end{align}
    where $\ket{v_m} \in M_{\Delta_3}$ and $\ket{v_0} = \ket{\Delta_3}$.
\end{definition}
If the Verma module $M_{\Delta_3}$ is irreducible, then a regular vertex operator exists uniquely, and the coefficients $\ket{v_k}$ are determined by the relations
\begin{align}\label{eq regular relation}
    L_{n} \ket{v_k}
    &= \left( \Delta_3 + n \Delta_2 - \Delta_1 + k - n + \delta_{n,0} \Delta_1 \right) \ket{v_{k-n}} \quad (n\geq 0).
\end{align}

A dual Verma module $M_\Delta^{*}$ with the highest weight $\Delta$ is the induced module
    \begin{align*}
        M_{\Delta}^{*} &= \mathrm{Ind}_{\mathrm{Vir}_{\leq 0}} ^{\mathrm{Vir}} \mathbb{C} \bra{\Delta},
    \end{align*}
    where $\mathrm{Vir}_{\leq 0}$ acts on the highest weight vector $\bra{\Delta}$ as
    \begin{equation*}
    \bra{\Delta} L_n=0\quad (n<0),\quad \bra{\Delta} L_0=\Delta\bra{\Delta}.
\end{equation*}
A dual regular vertex operator $\Phi_{\Delta_3,\Delta_1}^{\Delta_2, *}(z)\colon M_{\Delta_3}^* \to M_{\Delta_1}^*$ is defined in a manner similar to Definition \ref{def reg vo} by 
\begin{align*}
        [L_n,\Phi_{\Delta_3,\Delta_1}^{\Delta_2,*}(z)] &= z^n \left( z \dfrac{\partial}{\partial z} + (n+1) \Delta_2 \right) \Phi_{\Delta_3,\Delta_1}^{\Delta_2 ,*}(z), 
        \\
        \bra{\Delta_3}\Phi_{\Delta_3,\Delta_1}^{\Delta_2 ,*}(z)  
        &= z^{\Delta_3 - \Delta_2 -\Delta_1} \sum_{m \geq 0} \bra{v_m} z^{-m},
    \end{align*}
where $\bra{v_m}\in M_{\Delta_1}^*$ and $\bra{v_0}=\bra{\Delta_1}$. 

\begin{definition}
       A pairing $\braket{\ |\cdot|\ } : M_\Delta^{*} \times M_{\Delta} \to \mathbb{C}$ is defined by
    \begin{align*}
        &\bra{\Delta} \cdot \ket{\Delta} = 1, \quad
        \bra{u}L_n \cdot \ket{v} = \bra{u} \cdot L_n \ket{v} = \braket{u|L_n|v},
    \end{align*}
    where $\bra{u}\in M_{\Delta}^*$, $\ket{v}\in M_\Delta$.
\end{definition}

A regular conformal block with $n+2$ points is defined as the expectation value of the regular vertex operators
\begin{equation*}
    \bra{\Delta_{n+1}} \Phi_{\Delta_{n+1},\widetilde{\Delta}_{n-1}}^{\Delta_{n}}(z_n)\circ\cdots\circ  \Phi_{\widetilde{\Delta}_1,\Delta_0}^{\Delta_1}(z_1)\ket{\Delta_0}.
\end{equation*}

\subsection{Irregular vertex operators}
Let $\mathrm{Vir}_{\geq r}$ ($r\in\Z_{\geq 0}$) be the subalgebra generated by $L_n$ ($n\geq r$).
For a weight
\begin{equation*}
 \Lambda=(\Lambda_{r},\Lambda_{r+1},\ldots,\Lambda_{2r})\in\C^{r+1},
\end{equation*}
 an irregular Verma module $M_{\Lambda}^{[r]}$ of rank $r$ is an induced module
\begin{align*}
    M_{\Lambda} ^{[r]} &= \mathrm{Ind}_{\mathrm{Vir}_{\geq r}} ^{\mathrm{Vir}} \mathbb{C} \ket{\Lambda},
\end{align*}
where $\ket{\Lambda}$ is the irregular vector that satisfies
\begin{align*}
    L_n \ket{\Lambda} &= \Lambda_{n} \ket{\Lambda} \quad (n = r, r+1 , \ldots ,2r ), \qquad L_n \ket{\Lambda} = 0 \quad (n>2r).
\end{align*}
The condition $\Lambda_{2r} \neq 0$ implies that the irregular Verma module is irreducible \cite{FJK}.
\begin{definition}\label{def irr vo}
    Let $r$ be a positive integer, then an irregular vertex operator $\Phi_{\Lambda', \Lambda} ^{\Delta}(z) : M_{\Lambda} ^{[r]} \to M_{\Lambda'} ^{[r]} $ is defined by
    \begin{align}
        [L_n, \Phi_{\Lambda', \Lambda} ^{\Delta}(z)]
        &= z^n \left( z \frac{\partial}{\partial z} + (n+1) \Delta \right)  \Phi_{\Lambda', \Lambda} ^{\Delta}(z),
        \nonumber\\
        \Phi_{\Lambda', \Lambda}^{\Delta } (z) \ket{\Lambda}
        &= z^{\alpha} \exp{\left( \sum_{i=1} ^r \frac{\beta_i}{z^i} \right)} \sum_{m=0} ^{\infty} \ket{v_m} z^m,
        \label{eq irreg vo action}
    \end{align}
    where $\ket{v_m} \in M_{\Lambda'} ^{[r]}$ and $\ket{v_0} = \ket{\Lambda'}$.
\end{definition}

The definition implies that, for $n \geq r$,
\begin{align}
    \label{eq recursive relation for IVO}
    L_n \ket{v_m} &=\sum_{i=0}^{r} \delta_{n,i+r}\Lambda_{ i+r} \ket{v_m} - \sum_{i=1} ^r i \beta_i \ket{v_{m+i-n}} + (\alpha + (n+1)\Delta + m-n) \ket{v_{m-n}}.
\end{align}

\begin{theorem}[\cite{Nagoya 2015}]\label{thm Nagoya 2015}
If $\Lambda_{2r}\neq 0$, then an irregular vertex operator exists and is uniquely determined by the parameters $\Lambda$, $\Delta$, and $\beta_r$ with
    \begin{align*}
    \Lambda_n' &= \Lambda_n - \delta_{n,r} r \beta_r \quad (n= r,\ldots, 2r),
\end{align*}
and moreover $\alpha$, $\beta_k$ ($k=1,\ldots,r-1$), $\ket{v_m}$ are polynomials in $\Lambda_r,\ldots,\Lambda_{2r},\Lambda_{2r}^{-1}$, $\beta_r$, $\Delta$.
\end{theorem}
\begin{example}
    In the rank-one case, we have
    \begin{align*}
    \alpha &= - \frac{\beta_1 (\Lambda_1 - \beta_1)}{2 \Lambda_2} - 2 \Delta,
\end{align*}
and in the rank-two case, we have
\begin{align*}
    \alpha &= \frac{\beta_2 ( \Lambda_3 ^2 - 4 \Lambda_4 (\Lambda_2 - 3 \beta_2))}{4 \Lambda_4 ^2} - 3 \Delta,\quad
    \beta_1 = \frac{\beta_2 \Lambda_3}{\Lambda_4}.
\end{align*}
\end{example}

\begin{definition}\label{def irr pair}
   (i) A pairing $\braket{\ |\cdot |\ }: M_{\Delta}^{*} \times M_{\Lambda} ^{[1]} \to \mathbb{C} $ between a dual Verma module $M_{\Delta}^{*}$ and an irregular Verma module $M_{\Lambda} ^{[1]}$ of rank one  is defined by
    \begin{align*}
        &\bra{\Delta} \cdot \ket{\Lambda} = 1, \quad
        \bra{u}L_n \cdot \ket{v} = \bra{u} \cdot L_n \ket{v} = \braket{u|L_n|v}.
    \end{align*}
    where $\bra{u} \in M_{\Delta}^{*}$, $\ket{v} \in M_{\Lambda} ^{[1]}$.

    (ii)
    A pairing $\braket{\ |\cdot|\ }: V_{0}^{*} \times M_{\Lambda} ^{[2]} \to \mathbb{C} $ between the irreducible highest weight module $V_0^*$ and an irregular Verma module $M_{\Lambda} ^{[2]}$ of rank two
    \begin{align*}
        &\bra{0} \cdot \ket{\Lambda} = 1, \quad \bra{u}L_n \cdot \ket{v} = \bra{u} \cdot L_n \ket{v} = \braket{u|L_n|v}
    \end{align*}
    where $\bra{u} \in V_{0}^{*}$, $\ket{v} \in M_{\Lambda} ^{[2]}$.
\end{definition}

An irregular conformal block with $n+1$ regular singular points and one irregular singular point 0 of rank one is defined as the expectation value of the irregular vertex operators
\begin{equation*}
    \bra{\Delta}\cdot \Phi_{\Lambda^{(n)},\Lambda^{(n-1)}}^{\Delta_{n}}(z_n)\circ\cdots\circ \Phi_{\Lambda^{(1)},\Lambda}^{\Delta_1}(z_1)\ket{\Lambda}.
\end{equation*}

We emphasize that, unlike in the regular case, the vertex operator is not assumed to act on the dual module. The pairing in Definition 2.4 is defined directly between a dual Verma module and an irregular Verma module of rank one, and no dual action of the vertex operator is assumed here.

A dual irregular Verma module $M_{\Lambda}^{[r],*}$ of rank $r$ with the weight $\Lambda$ is the induced module
\begin{align*}
    M_{\Lambda} ^{[r],*} &= \mathrm{Ind}_{\mathrm{Vir}_{\leq -r}} ^{\mathrm{Vir}} \mathbb{C} \bra{\Lambda},
\end{align*}
where $\bra{\Lambda}$ is the dual irregular vector that satisfies
\begin{align*}
     \bra{\Lambda}L_{-n} &= \bra{\Lambda}\Lambda_{-n}  \quad (n = r, r+1 , \ldots ,2r ), \qquad  \bra{\Lambda}L_n = 0 \quad (n<-2r).
\end{align*}
A dual irregular vertex operator $\Phi_{\Lambda, \Lambda'} ^{\Delta,*}(z) : M_{\Lambda} ^{[r],*} \to M_{\Lambda'} ^{[r],*} $ is defined in a similar way to Definition \ref{def irr vo}, and pairings between $M_\Lambda^{[1],*} \times M_\Delta$ and $M_\Lambda^{[2],*}\times V_0$ are also defined in a similar way to Definition \ref{def irr pair}.  These dual irregular vertex operators will be used in Section \ref{sec tau functions} to construct the conformal blocks for the $\tau$ function expansion.

\subsection{Free field representation}\label{sec free field representation}

The Heisenberg algebra $\mathrm{H}$ is the Lie algebra spanned by $a_n (n \in \mathbb{Z})$, $q$ and $\mathbf{1}$ with commutation relations
\begin{align*}
    [a_m, a_n] &= m \delta_{m+n,0}\mathbf{1}, \quad
    [a_m, q] = \delta_{m,0}\mathbf{1}, \quad
    [a_m, \mathbf{1}] = 0, \quad
    [q,\mathbf{1}]=0.
\end{align*}
Let $\mathrm{H}_{\geq 0}$ be a subalgebra generated by $a_n$ ($n\geq 0$).
\begin{definition}
    For the vector $\ket{\lambda}$ ($\lambda\in\C$) satisfying the condition
    \begin{equation*}
     a_0 \ket{\lambda} = \lambda \ket{\lambda}, \quad a_n \ket{\lambda}=0\quad
     (n\geq 1),
    \end{equation*}
     a Fock space $F_{\lambda}$ is defined as induced module
    \begin{align*}
        F_{\lambda} &= \mathrm{Ind}_{\mathrm{H}_{\geq 0}} ^{\mathrm{H}} \mathbb{C} \ket{\lambda}.
    \end{align*}
\end{definition}

\begin{definition}
    The normal order is defined as
    \begin{align*}
        :a_m a_n: =
        \begin{cases}
            a_m a_n \; (n \geq m), \\
            a_n a_m \; (n < m),
        \end{cases}
        :q a_n: =
        \begin{cases}
            q a_n \; (n \geq 0), \\
            a_n q \; (n < 0).
        \end{cases}
    \end{align*}
\end{definition}

A free field realization of the Virasoro algebra is given by
\begin{align*}
    L_n = \frac{1}{2} \sum_{k \in \mathbb{Z}} :a_{n-k} a_k: - \rho (n+1) a_n.
\end{align*}
With this realization, the  $L_n$ satisfy the commutation relations of the Virasoro algebra on each Fock space $F_\lambda$ with central charge
$c=1-12\rho^2$. Then for generic $\lambda$, the Fock space $F_\lambda$ is
isomorphic to the Verma module $M_{\Delta_{\lambda}}$, where
\begin{align*}
    \Delta(\lambda) &= \frac{1}{2} \lambda (\lambda - 2 \rho).
\end{align*}

A formal sum $\exp\left(\sum_{k=1}^r \frac{\lambda_{k}}{k}a_{-k}\right)\ket{\lambda_0}$ is an irregular vector of rank $r$ provided that $\lambda_r\neq 0$. The parameters $\Lambda_n$ are given by
\begin{equation}\label{eq Lambda n}
\Lambda_n = \frac{1}{2} \sum_{k=n-r} ^{r} \lambda_{n-k} \lambda_k -\delta_{n,r}(r+1)\rho \lambda_r
\quad (n=r,r+1,\ldots, 2r).
\end{equation}
Since the condition $\lambda_r \neq 0$ implies $\Lambda_{2r} \neq 0$, which ensures that the irregular Verma module $M_\Lambda^{[r]}$ is irreducible \cite{FJK}, the space $F_\lambda^{[r]}$ constructed from this irregular vector is isomorphic to $M_\Lambda^{[r]}$.

The free boson field $\varphi (z)$ is defined by
\begin{align*}
    \varphi(z) &= q + a_0 \log{z} - \sum_{n \neq 0} \frac{a_{n}}{n} z^{-n}.
\end{align*}

Let us recall that an operator
\begin{align*}
    :e^{\lambda_2 \varphi(z)}:
\end{align*}
obtained by exponentiating the free boson field is
a vertex operator $\Phi_{\Delta_3,\Delta_1}^{\Delta_2}(z) : M_{\Delta_1} \to M_{\Delta_3}$ with
\begin{align*}
    \Delta_{1} &= \Delta (\lambda_1), \quad \Delta_2 = \Delta(\lambda_2) ,\quad \Delta_3 = \Delta(\lambda_1 + \lambda_2).
\end{align*}
Furthermore, the operator $:e^{\lambda_z \varphi(z)}:$ is also regarded as an irregular vertex operator $\Phi_{\Lambda', \Lambda} ^{\Delta}(z) : M_{\Lambda} ^{[r]} \to M_{\Lambda'} ^{[r]} $ with the same parametrization of \(\Lambda_n\) as above \eqref{eq Lambda n}. In this case,
\begin{align*}
\Delta=\Delta(\lambda_z),\quad \alpha=\lambda_0\lambda_z, \quad \beta_k=-\frac{\lambda_k\lambda_z}{k} \quad (k=1,\ldots, r).
\end{align*}

To obtain more general concrete vertex operators, we consider
the screening operator $Q_{+}$ defined by
\begin{align*}
    Q_{+} &= \int_\gamma :e^{\lambda_{+} \varphi(t)}: dt,
\end{align*}
where the screening charge $\lambda_+$ is fixed as a root of the equation $\rho = \frac{\lambda_+}{2} - \frac{1}{\lambda_+}$. The superscript $+$ indicates this choice of screening charge.

If we choose an integral path $\gamma$ appropriately, then
the screening operator $Q_{+}$ commutes with the generators $L_n$ of the Virasoro algebra.
In what follows, \(Q_+^n\) is used as shorthand for an \(n\)-fold insertion of the screening current.
Whenever \(Q_+^n\) appears together with a vertex operator acting on a Fock highest vector or irregular vector $v$ considered below, the expression
\begin{equation*}
    :e^{\lambda\varphi(z)}:Q_+^n v
\end{equation*}
is interpreted coefficientwise as the integral obtained from
\begin{equation*}
    \int_\gamma dt_1\cdots dt_n\,
:e^{\lambda\varphi(z)}::e^{\lambda_{+} \varphi(t_1)}:\cdots :e^{\lambda_{+} \varphi(t_n)}:v
\end{equation*}
with the appropriate integration path \(\gamma\); that is, after this vector-valued integral is expanded at \(z=0\), asymptotically in the irregular case, its vector coefficients are written in the Heisenberg monomial basis with scalar coefficients given by convergent integrals.

In the following two propositions, we assume that the variables and parameters are positive real numbers in order to fix branches and contours explicitly. This assumption is not essential for the formal expansions. If $z=|z|e^{\sqrt{-1}\theta}$ with a fixed choice of $\arg z=\theta$, the contours may be chosen by analytic continuation from the positive real case.

For Proposition 2.1, we take
\begin{equation*}
    t_i=zx_i,\quad 0<x_1<\cdots<x_n<1.
\end{equation*}
Thus all variables $t_i$ lie on the segment from $0$ to $z$.

For Proposition 2.2, we keep the change of variables
\begin{equation*}
    t_i=\frac{z}{1-z^r s_i}.
\end{equation*}
If
\begin{equation*}
    s_i=x_i e^{-r\sqrt{-1}\theta},
\quad x_1<\cdots<x_n<0,
\end{equation*}
then
\begin{equation*}
    t_i=\frac{z}{1-z^r s_i}
\end{equation*}
lies on the same ray as $z$. With this choice, the expansion below is still written in powers of $z$.

\begin{proposition}\label{prop free field RVO}
 Let $z$ and $\lambda_k$ ($k=1,2,+)$ be positive real numbers. Let $\gamma=\{(t_1,\dots,t_n)\in\mathbb{R}^n \mid 0<t_1<\cdots<t_n<z\}$ be the integration domain. Then, the expansion of the integral representation $:e^{\lambda_2 \varphi(z)}: Q_{+}^n$
around $z=0$ yields a vertex operator $\Phi_{\Delta_3,\Delta_1}^{\Delta_2}(z) : M_{\Delta_1} \to M_{\Delta_3}$ with
\begin{align*}
    \Delta_{1} &= \Delta (\lambda_1), \quad \Delta_2 = \Delta(\lambda_2) ,\quad \Delta_3 = \Delta(\lambda_1 + \lambda_2 + n \lambda_{+}).
\end{align*}
\end{proposition}
\begin{proof}
 Because the screening operator $Q_+$ commutes with the Virasoro generators $L_n$, the composition $:e^{\lambda_2 \varphi(z)}: Q_{+} ^n$ satisfies the commutation relations given in \eqref{eq reg vo}. First, we explicitly write the action of the integral representation on the vector $\ket{\lambda_1}$:
    \begin{align*}
    :e^{\lambda_2 \varphi(z)}: Q_{+}^n \ket{\lambda_1}
    =&z^{\lambda_1\lambda_2}\int_\gamma dt\prod_{1 \leq i<j \leq n} (t_i - t_j)^{\lambda_+ ^2}\prod_{i=1}^n\left[
    (z-t_i)^{\lambda_2\lambda_+}t_i^{\lambda_1\lambda_+}
    \exp\left(\lambda_+\sum_{k=1}^\infty\frac{a_{-k}}{k}t_i^k\right)\right]
    \\
    &\times \exp\left(\lambda_2\sum_{k=1}^\infty\frac{a_{-k}}{k}z^{k}\right)
    \ket{\lambda_1+\lambda_2+n\lambda_+},
\end{align*}
    where $dt = dt_1\cdots dt_n$. To expand this expression around $z=0$, we make the change of variables $t_i=zs_i$.
    Then, we have
    \begin{align*}
    :e^{\lambda_2 \varphi(z)}: Q_{+}^n \ket{\lambda_1}=&
        z^{\lambda_1 \lambda_2 + n \lambda_+ (\lambda_1 + \lambda_2) + \frac{n(n-1)}{2} \lambda_+^2 + n}
        \int_{\gamma'} ds\prod_{1 \leq i<j \leq n} (s_i - s_j)^{\lambda_+ ^2}\prod_{i=1}^n
    (1-s_i)^{\lambda_2\lambda_+}s_i^{\lambda_1\lambda_+}
    \\
    &\times \exp\left(\lambda_+\sum_{i=1}^n\sum_{k=1}^\infty\frac{a_{-k}}{k}z^ks_i^k\right)\exp\left(\lambda_2\sum_{k=1}^\infty\frac{a_{-k}}{k}z^{k}\right)
    \ket{\lambda_1+\lambda_2+n\lambda_+},
    \end{align*}
    where $ds= ds_1 \cdots ds_n$. Thus, the expansion above is of the form \eqref{eq reg vo action}, and we obtain the relation

    \begin{equation*}
    \Delta_3
=
\Delta(\lambda_1)+\Delta(\lambda_2)
+\lambda_1\lambda_2
+n\lambda_+(\lambda_1+\lambda_2)
+\frac{n(n-1)}{2}\lambda_+^2+n
=
\Delta(\lambda_1+\lambda_2+n\lambda_+).
\end{equation*}
\end{proof}
We note that the integral formulas that appeared in the proof are the Selberg integrals, whose integration domains are
well known \cite{Aomoto Kita}.

Similarly, the following proposition shows that the asymptotic expansion of $:e^{\lambda_z \varphi(z)}: Q_{+}^n$ yields an irregular vertex operator in the sense of Definition \ref{def irr vo}.
\begin{proposition}\label{prop free field IVO}
 Let $z$ and $\lambda_k$ ($k=0,1,\ldots,r,z,+)$ be positive real numbers such that $0<z<1$. Let $\gamma = \{(t_1,\dots,t_n)\in\mathbb{R}^n \mid 0<t_1<\cdots<t_n<z\}$ be the integration domain. Then, the asymptotic expansion of the integral representation $:e^{\lambda_z \varphi(z)}: Q_{+}^n$ at $z=0$ is identified with the irregular vertex operator $\Phi_{\Lambda', \Lambda} ^{\Delta_2}(z) : M_{\Lambda} ^{[r]} \to M_{\Lambda'} ^{[r]}$, where the parameters $\Lambda_n$ are given by \eqref{eq Lambda n}, $\Lambda'$ is as in Theorem \ref{thm Nagoya 2015},  and the conformal dimension and parameters are specified by:
\begin{align*}
 &\Delta_2=\Delta(\lambda_z),\quad\\
&\alpha=(\lambda_z+n\lambda_+)\lambda_0+n(r+1)(\lambda_+\lambda_z+1) +\frac{n(n-1)(r+1)}{2}\lambda_+^2,
\\
&\beta_k= -\frac{\lambda_k(\lambda_z+n\lambda_+)}{k}\quad (k=1,2,\ldots, r).
\end{align*}
\end{proposition}
\begin{proof}
First, we explicitly write the action of the integral representation on the vector
$\exp{\left(\sum_{k=1}^{r}\frac{\lambda_k}{k} a_{-k} \right)} \ket{\lambda_0}$:
\begin{align*}
  &  :e^{\lambda_z \varphi(z)}: Q_{+} ^n
    \exp{\left(\sum_{k=1}^{r}\frac{\lambda_k}{k} a_{-k} \right)} \ket{\lambda_0}
\\
&= z^{\lambda_z \lambda_0} e^{-\sum_{k=1}^{r}\frac{\lambda_k \lambda_z }{kz^k}}\int_\gamma dt\prod_{1 \leq i<j \leq n} (t_i - t_j)^{\lambda_+ ^2}\prod_{i=1}^n\left[
    (z-t_i)^{\lambda_z\lambda_+}t_i^{\lambda_0\lambda_+}
    \exp\left(-\sum_{k=1}^{r}\frac{\lambda_k \lambda_+ }{kt_i^k}\right)\right]
    \\
    &\times \exp\left(\lambda_+\sum_{i=1}^n\sum_{k=1}^\infty\frac{a_{-k}}{k}t_i^k\right)\exp\left(\lambda_z\sum_{k=1}^\infty\frac{a_{-k}}{k}z^{k}\right)
    \exp{\left(\sum_{k=1}^{r}\frac{\lambda_k}{k} a_{-k} \right)}\ket{\lambda_0+\lambda_z+n\lambda_+}.
\end{align*}
To expand this expression around $z=0$, we make the change of variables $t_i = \frac{z}{1- z^r s_i}$.
    This yields
    \begin{align*}
    &  :e^{\lambda_z \varphi(z)}: Q_{+} ^n
    \exp{\left(\sum_{k=1}^{r}\frac{\lambda_k}{k} a_{-k} \right)} \ket{\lambda_0}
\\
   &=  (-1)^{n\lambda_+ \lambda_z}z^{(\lambda_z + n\lambda_+) \lambda_0 + n(r+1)(\lambda_+ \lambda_z+1) +n(n-1)(r+1)\lambda_+^2/2}
    \exp{\left( \sum_{k=1} ^{r} \frac{-\lambda_k (\lambda_z+n\lambda_+)}{k z^k} \right)}
    \\
    & \times \int_{\gamma'} ds \prod_{1 \leq i<j \leq n} (s_i - s_j)^{\lambda_+ ^2}\prod_{i=1}^n s_i^{\lambda_+ \lambda_z} \left( 1 - z^r s_i \right)^{ (1-n)\lambda_+^2-\lambda_+ (\lambda_0 + \lambda_z) -2}
    \\
    &\times \exp\left( \sum_{i=1}^n\sum_{k=1} ^{r} -\frac{\lambda_k \lambda_+}{k} \left( \left(\frac{1- z^r s_i}{z}  \right)^k  - \frac{1}{z^k}\right)\right)\exp{\left( \lambda_{+} \sum_{i=1}^n\sum_{k=1} ^{\infty} \frac{a_{-k}}{k}\left( \frac{z}{1- z^r s_i} \right)^k \right)}
   \\
   & \times  \exp\left( \lambda_z \sum_{k=1} ^{\infty} \frac{a_{-k}}{k}z^k \right)
    \exp{\left( \sum_{k=1} ^{r} \frac{\lambda_k}{k} a_{-k} \right)} \ket{\lambda_0+ \lambda_z + n\lambda_+},
\end{align*}
where the integration domain $\gamma'$ is equal to $\{(s_1,\ldots,s_n)\in\mathbb{R}^n\mid s_1<\cdots<s_n<0\}$.

Observe that
\begin{equation*}
    \left(\frac{1- z^r s_i}{z}  \right)^k  - \frac{1}{z^k}
    = z^{-k}\sum_{j=1}^k \binom{k}{j}(-z^rs_i)^j
    \quad (k=1,\ldots,r).
\end{equation*}
For $k=r$, the $j=1$ term in the summation yields $-r s_i$. This term contributes to the converging factor $e^{\lambda_r\lambda_+ s_i}$ in the integral. Consequently, we need to evaluate the asymptotic expansion of an integral of the form
\begin{align*}
&\int_{\gamma'} ds \prod_{1 \leq i<j \leq n} (s_i - s_j)^{\lambda_+ ^2}\prod_{i=1}^n s_i^{a} \left( 1 - z^r s_i \right)^{b} e^{\lambda_r\lambda_+ s_i}
    \\
    &\times \prod_{i=1}^n\exp\left( \sum_{k=1} ^{r-1} -\frac{\lambda_k \lambda_+}{k} z^{-k}\sum_{j=1}^k \binom{k}{j}(-z^rs_i)^j-\frac{\lambda_r \lambda_+}{r}z^{-r}\sum_{j=2}^r \binom{r}{j}(-z^rs_i)^j\right),
    \end{align*}
where $a>0,b\in \mathbb{R}$.

Next, we expand the terms in the integrand as
\begin{align*}
    &(1-z^rs_i)^b=\sum_{j=0}^N\binom{b}{j}(-z^rs_i)^j+R_{i,N},
    \\
    &\exp\left(-\frac{\lambda_k \lambda_+}{k} z^{-k}\binom{k}{j}(-z^rs_i)^j\right)=
    \sum_{l=0}^{N_{k,i,j}}\frac{1}{l!}\left(-\frac{\lambda_k \lambda_+}{k} z^{-k}(-z^rs_i)^j\right)^l+R_{k,i,j,N_{k,i,j}},
\end{align*}
for integers $N,N_{k,i,j}\geq 0$. By Taylor's theorem, the remainder satisfies the bound $|R_{k,i,j,N_{k,i,j}}|\leq K_1|z^{-k}(z^rs_i)^j|^{N_{k,i,j}+1}$ for some constant $K_1>0$ on the integration domain $\gamma'$. In order to estimate $R_{i,N}$, we decompose $\gamma'$ as $\bigcup_{\ell=0}^n\gamma'_\ell$, where
\begin{align*}
        \gamma'_\ell=\left\{(s_1,\ldots,s_n)\in\mathbb{R}^n\;  \middle|\; s_1 < \cdots < s_\ell < -\frac{1}{rz^r} \leq s_{\ell+1} < \cdots < s_n < 0 \right\}.
\end{align*}
Using the inequalities $|R_{i,N}|\leq K_2|z^rs_i|^{N+1}$ for $-\frac{1}{rz^r} \leq s_i < 0$, and $|R_{i,N}|\leq K_3 |z|^{-m_1}|s_i|^{m_2}$ for $s_i < -\frac{1}{rz^r}$, with some constants $K_2,K_3>0$ and integers $m_1,m_2>0$,
we obtain
\begin{align*}
    \left|\int_{\gamma'} ds \prod_{1 \leq i<j \leq n} (s_i - s_j)^{\lambda_+ ^2}\prod_{i=1}^n s_i^{a} e^{\lambda_r\lambda_+ s_i}R_{i,N}\right|\leq K_4 |z|^{rN+1}
\end{align*}
for some $K_4>0$. On the region \(s_i<-1/(r|z|^r)\), the remainder is bounded by a polynomial in
\(|s_i|\) and a negative power of \(|z|\). Since \(\lambda_r\lambda_+>0\),
the factor \(e^{\lambda_r\lambda_+s_i}\) gives exponential decay on this tail.
Hence the tail contribution is \(O(e^{-c/|z|^r}|z|^{-M})\), and therefore is
\(O(|z|^{rN+1})\) for any fixed \(N\).

Therefore, the asymptotic expansion of the integral representation
\begin{equation*}
     :e^{\lambda_z \varphi(z)}: Q_{+} ^n
    \exp{\left(\sum_{k=1}^{r}\frac{\lambda_k}{k} a_{-k} \right)} \ket{\lambda_0}
\end{equation*}
takes the form given in \eqref{eq irreg vo action}. Finally, the parameters $\beta_k$ can be directly read off from this integral representation, which completes the proof.

\end{proof}

As established in Theorem \ref{thm Nagoya 2015}, the irregular vertex operator is uniquely determined by $\Lambda$, $\Delta$, and $\beta_r$. Therefore, the parameters $\alpha$ and $\beta_k$ ($k=1,\dots,r-1$) are naturally expressed as functions of $\beta_r$.

\begin{corollary}
    For an irregular vertex operator $\Phi_{\Lambda', \Lambda} ^{\Delta}(z) : M_{\Lambda} ^{[r]} \to M_{\Lambda'} ^{[r]} $ parametrized by $\Lambda_n = \frac{1}{2} \sum_{k=n-r} ^{r} \lambda_{n-k} \lambda_k -\delta_{n,r}(r+1)\rho \lambda_r$ ($n=r,r+1,\ldots, 2r$), we have
    \begin{align}
    &\alpha=\frac{(r+1)}{2}\left(r\frac{\beta_r}{\lambda_r}+\lambda_z\right)
    \left(r\frac{\beta_r}{\lambda_r}-\lambda_z+2\rho\right)
    -r\frac{\lambda_0\beta_r}{\lambda_r},
    \label{eq alpha}
    \\
    &\beta_k=\frac{r\lambda_k}{k\lambda_r}\beta_r\quad (k=1,\ldots,r-1). \label{eq beta}
    \end{align}
\end{corollary}
\begin{proof}
    From Theorem \ref{thm Nagoya 2015}, $\alpha$, $\beta_k$ ($k=1,\ldots,r-1$) are polynomials in $\lambda_0, \lambda_1,\ldots,\lambda_r$, $\lambda_r^{-1}$, $\lambda_z$, $\beta_r$ and $\rho$. From Proposition \ref{prop free field IVO}, the relations \eqref{eq alpha} and \eqref{eq beta} hold for infinitely many values of $\beta_r$.
Hence, they hold identically as polynomial identities in $\beta_r$. This completes the proof.
\end{proof}

\section{Degeneration limit of vertex operators}
In this section, we prove the degeneration limits of the vertex operators stated in Theorem \ref{main result 1-1} and Theorem \ref{main result 1-2} by combining the rearranged expansions with the free field integral representations reviewed in Section \ref{sec free field representation}.

The rearranged expansions in this section are understood as formal expansions in the variable $w$. The coefficients $\ket{R_k(z)}$ are defined algebraically by equating the coefficients of $w^k$. Thus the problem in Theorems \ref{Thm RVO to IVO} and \ref{Thm IVO to IVO} is the coefficientwise convergence of $\ket{R_k(z)}$ under the scaling limits. In the proofs below, we take $z$ and $w$ to be positive real numbers only in order to choose explicit contours for the coefficient integrals. The statements of the theorems do not  require $z$ and $w$ to be real variables.

\subsection{Regular vertex operator}\label{subsec rvo degeneration}

    Consider the composition of two regular vertex operators $\Phi_{\Delta_5, \Delta_3}^{\Delta_4}(z)\Phi_{\Delta_3, \Delta_1}^{\Delta_2}(w): M_{\Delta_1}\to M_{\Delta_5}$. Following Gaiotto and Teschner (see Appendix D.2.1, Eq. (D.6) of \cite{GT}), we use the expansion
\begin{equation}\label{eq rank 0 to 1 re ex}
\Phi_{\Delta_5, \Delta_3}^{\Delta_4}(z)\Phi_{\Delta_3, \Delta_1}^{\Delta_2}(w)\ket{\Delta_1}
=z^{\Delta_5-\Delta_4-\Delta_3}w^{\Delta_3-\Delta_2-\Delta_1}\left(1-\frac{w}{z}\right)^A\sum_{k=0}^{\infty}\ket{R_k(z)}w^k,
\end{equation}
where the factor $\left(1-w/z\right)^A$ is  expanded in powers
of \(w\), namely
\begin{equation*}
    \left(1-\frac{w}{z}\right)^A
    =
    \sum_{n=0}^{\infty}\frac{A(A-1) \cdots (A-n+1)}{n!} \left( -\frac{w}{z}\right)^n .
\end{equation*}
The exponent $A$ is chosen in the degeneration limit below so that the divergent part of the collision is absorbed by the factor $\left(1-w/z\right)^A$ and the coefficients $\ket{R_k(z)}$ admit finite limits. The coefficients $ \ket{R_k(z)}$ are then defined algebraically by equating the coefficients of $w^k$ on both sides.

\begin{theorem}\label{Thm RVO to IVO}
The coefficients \(\ket{R_k(z)}\) converge
coefficientwise to the coefficients of
\[
\Phi_{\Lambda',\Lambda}^{\Delta}(w)\ket{\Lambda}, 
\] 
in the limit
\begin{align}
    \lambda_1 &= -\frac{c_1}{\epsilon} + \frac{c_0}{2}, \qquad
    \lambda_3 = -\frac{c_1}{\epsilon} + \frac{c_0}{2} + \beta, \qquad
    \lambda_4 = \frac{c_1}{\epsilon} + \frac{c_0}{2}, \label{eq limit rvo to ivo 1}\\
    A &= \frac{c_1\beta}{\epsilon} + \frac{1}{2}\beta(c_0-2\rho+\beta)-\Delta, \label{eq limit rvo to ivo 2}\\
    \Delta_2 &= \Delta, \qquad z=\epsilon, \qquad \epsilon\to 0 . \label{eq limit rvo to ivo 3}
\end{align}
 Here $c_0$ and $c_1$ are the parameters of the limiting rank-one
irregular vector in the parametrization given by \eqref{eq Lambda n}. Namely,

\[
(\Lambda_1,\Lambda_2)=\bigl(c_1(c_0-2\rho),\, c_1^2/2\bigr),\qquad
\Lambda_1'=\Lambda_1+c_1\beta,\qquad \Lambda_2'=\Lambda_2,
\]
and $\Delta_i=\Delta(\lambda_i)$ for $i=1,3,4$.
\end{theorem}

\begin{lemma}\label{lem recursive limit rvo to ivo}
The coefficients $\ket{R_k(z)}$ $(k\geq 0)$ of the rearranged expansion \eqref{eq rank 0 to 1 re ex} satisfy the recursive relations
\begin{align*}
    L_n \ket{R_k(z)}
    =&\,
    z^n\left(\Delta_5+n\Delta_4-\Delta_3
    +z\frac{\partial}{\partial z}\right)\ket{R_k(z)}
    +A\sum_{m=1}^{n-1} z^{n-m}\ket{R_{k-m}(z)} \\
    &\,
    +\bigl(A+\Delta_3+n\Delta_2-\Delta_1+k-n\bigr)
    \ket{R_{k-n}(z)}
    \qquad (n\geq 0),
\end{align*}
where $\ket{R_j(z)}=0$ for $j<0$. Assume that $\ket{R_0(z)}\to\ket{\Lambda'}$ and that all $\ket{R_k(z)}$ converge to $\ket{u_k}$ under the limit \eqref{eq limit rvo to ivo 1}--\eqref{eq limit rvo to ivo 3}. Then
\begin{align*}
    L_n \ket{u_k}
    =&\,
    (\delta_{n,1}\Lambda_1+\delta_{n,2}\Lambda_2)\ket{u_k}
    +c_1\beta\,\ket{u_{k-n+1}} \\
    &\,
    +\bigl(\alpha+(n+1)\Delta+k-n\bigr)\ket{u_{k-n}}
    \qquad (n\geq 1),
\end{align*}
where $\ket{u_j}=0$ for $j<0$ and $\alpha$ is the rank-one value in \eqref{eq alpha}. Consequently, if the coefficients $\ket{R_k(z)}$ converge, then their limits coincide with the coefficients of $\Phi_{\Lambda',\Lambda}^{\Delta}(w)$.
\end{lemma}
\begin{remark}
    This Lemma \ref{lem recursive limit rvo to ivo} was discussed in \cite{Lisovyy Nagoya Roussillon}
\end{remark}
\begin{proof}
The first relation follows from the commutation relations \eqref{eq reg vo} and definition of the rearranged expansion \eqref{eq rank 0 to 1 re ex}. Taking the limit \eqref{eq limit rvo to ivo 1}--\eqref{eq limit rvo to ivo 3}, we obtain the second relation. Comparing it with \eqref{eq recursive relation for IVO}, we find that the limiting coefficients satisfy the $r=1$ specialization of the defining recursive relations for the irregular vertex operator $\Phi_{\Lambda',\Lambda}^{\Delta}(w)$, with $\beta_1=-c_1\beta$. Since $\ket{R_0(z)}\to\ket{\Lambda'}$, the assertion follows from the uniqueness of the irregular vertex operator in Theorem~\ref{thm Nagoya 2015}.
\end{proof}

\begin{proof}

First, we analyze the limit of $ \ket{R_0(z)}$. By definition,
\begin{equation*}
     \ket{R_0(z)}
    = z^{\Delta_3+\Delta_4-\Delta_5}\Phi_{\Delta_5,\Delta_3}^{\Delta_4}(z)\ket{\Delta_3}
    = \sum_{k=0}^{\infty}\ket{v_k}z^k
    \qquad (\ket{v_k}\in M_{\Delta_5}).
\end{equation*}
By the recursive relations \eqref{eq regular relation} for $\ket{v_k}$, the vector $ \ket{R_0(z)}$ converges to $\ket{\Lambda'}$ in the limit
\begin{align*}
    \Delta_4-\Delta_3 = \frac{\Lambda_1'}{\epsilon}+O(1),\quad
    2\Delta_4-\Delta_3 = \frac{\Lambda_2'}{\epsilon^2}+O(\epsilon^{-1}),\quad
    z=\epsilon,\quad \epsilon\to 0.
\end{align*}

By Lemma~\ref{lem recursive limit rvo to ivo}, it remains to prove the convergence of the coefficients $\ket{R_k(z)}$ under the limit \eqref{eq limit rvo to ivo 1}--\eqref{eq limit rvo to ivo 3}.
The coefficient $ \ket{R_k(z)}$ is a rational function of $\Delta_3$ and a polynomial in $\Delta_i$ for $i=1,2,4,5$. After substituting \eqref{eq limit rvo to ivo 1}--\eqref{eq limit rvo to ivo 3} into $ \ket{R_k(z)}$, it becomes a Laurent series in $\epsilon$ whose coefficients are polynomials in $\beta$. Hence, it suffices to prove the convergence of $ \ket{R_k(z)}$ for infinitely many values of $\beta$.

We also note that $\Delta_5$ may be chosen arbitrarily. The point is that, although $\Delta_5$ appears in the recursive relations for $ \ket{R_k(z)}$, it does not appear in the expression of $ \ket{R_k(z)}$ in terms of Virasoro generators acting on $ \ket{R_0(z)}$. In other words, the dependence on $\Delta_5$ is carried entirely by $ \ket{R_0(z)}$. Since $ \ket{R_0(z)}$ converges to the irregular vector in the limit of the theorem, the convergence of $ \ket{R_k(z)}$ is unaffected by the choice of $\Delta_5$. Although the free field realization below restricts the out-state to the momentum
\[
\lambda_5=\lambda_1+\lambda_2+\lambda_4+n\lambda_+,
\]
and hence fixes $\Delta_5$, this is only a convenient choice for proving convergence for infinitely many values of $\beta$, and does not affect the general statement of the theorem.

To prove the required convergence, we now use the free field realization. The composition of two regular vertex operators has the form
\begin{align*}
    &:e^{\lambda_4 \varphi (z)}::e^{\lambda_2 \varphi (w)}: Q_{+}^n \ket{\lambda_1} \\
    &=
    z^{\lambda_4 \lambda_1}
    w^{\lambda_2 \lambda_1}
    (z-w)^{\lambda_4 \lambda_2}
    \int_{\gamma} dt
    \prod_{1 \leq i<j \leq n} (t_i - t_j)^{\lambda_+^2}
    \prod_{i=1}^{n}
    t_i^{\lambda_+ \lambda_1}
    (w-t_i)^{\lambda_2 \lambda_+}
    (z-t_i)^{\lambda_4 \lambda_+} \\
    &\qquad \times
    \exp\left(\lambda_+ \sum_{i=1}^n \sum_{k=1}^{\infty} \frac{a_{-k}}{k} t_i^k \right)
    \exp\left(\lambda_2 \sum_{k=1}^{\infty} \frac{a_{-k}}{k} w^k
    +\lambda_4 \sum_{k=1}^{\infty} \frac{a_{-k}}{k} z^k \right)
    \ket{\lambda_1+\lambda_2+\lambda_4+n\lambda_+}.
\end{align*}
 For the purpose of evaluating the coefficient integrals in the free field representation, we take $z$ and $w$ to be positive real numbers with $z>w>0$, and choose
\begin{equation*}
    \gamma=\{(t_1,\ldots,t_n)\mid w>t_1>\cdots>t_n>0\}.
\end{equation*}
We also take the parameters $\lambda_k$ $(k=1,2,+)$ to be positive real numbers in order to fix branches and ensure convergence of the integral. Since $\lambda_3=\lambda_1+\lambda_2+n\lambda_+$, the parameter $\beta$ in \eqref{eq limit rvo to ivo 1} and \eqref{eq limit rvo to ivo 2} is equal to $\lambda_2+n\lambda_+$.

To expand the integral in powers of $w$, we make the change of variables
\begin{align*}
    t_i=\frac{w(1-zs_i)}{1-ws_i}.
\end{align*}
Then the above expression becomes
\begin{align*}
    &:e^{\lambda_4 \varphi (z)}::e^{\lambda_2 \varphi (w)}: Q_{+}^n \ket{\lambda_1} \\
    =&\,
    (-1)^{-\frac{n}{2}(n-1)\lambda_+^2}
    z^{n\lambda_+\lambda_2+\frac{n}{2}(n-1)\lambda_+^2+n}
    C_{0,1}
    \int_{\gamma'} ds
    \prod_{1 \leq i<j \leq n} (s_i-s_j)^{\lambda_+^2} \\
    &\times
    \prod_{i=1}^{n}
    s_i^{\lambda_+\lambda_2}
    (1-ws_i)^{-\lambda_+(\lambda_1+\lambda_4+\lambda_2)-(n-1)\lambda_+^2-2}
    (1-zs_i)^{\lambda_+\lambda_1} \\
    &\times
    \exp\left(\lambda_+ \sum_{i=1}^n \sum_{k=1}^{\infty}
    \frac{a_{-k}}{k}\left(\frac{w(1-zs_i)}{1-ws_i}\right)^k\right)
    \exp\left(\lambda_2 \sum_{k=1}^{\infty}\frac{a_{-k}}{k} w^k\right) \\
    &\times
    \exp\left(\lambda_4 \sum_{k=1}^{\infty}\frac{a_{-k}}{k} z^k\right)
    \ket{\lambda_1+\lambda_2+\lambda_4+n\lambda_+},
\end{align*}
where
\begin{align*}
    C_{0,1}
    =
    z^{(\lambda_1+\lambda_2+n\lambda_+)\lambda_4}
    w^{(\lambda_2+n\lambda_+)\lambda_1+n\lambda_2\lambda_+
    +\frac{n}{2}(n-1)\lambda_+^2+n}
    \left(1-\frac{w}{z}\right)^{(\lambda_2+n\lambda_+)\lambda_4
    +n\lambda_+\lambda_2+\frac{n}{2}(n-1)\lambda_+^2+n},
\end{align*}
and
\[
\gamma'=\{(s_1,\ldots,s_n)\mid 0<s_1<\cdots<s_n<1/z\}.
\]

In this free field representation, we have
\begin{align*}
    \Delta_5-\Delta_4-\Delta_3
    &= (\lambda_1+\lambda_2+n\lambda_+)\lambda_4,\\
    \Delta_3-\Delta_2-\Delta_1
    &= (\lambda_2+n\lambda_+)\lambda_1+n\lambda_2\lambda_+
      +\frac{n}{2}(n-1)\lambda_+^2+n,\\
    A
    &= (\lambda_2+n\lambda_+)\lambda_4+n\lambda_+\lambda_2
      +\frac{n}{2}(n-1)\lambda_+^2+n.
\end{align*}
Expanding the integral in powers of $w$, we obtain
\begin{align*}
    C_{0,1}
    \sum_{k=0}^{\infty}\sum_{|\nu|\leq k}
    r_{\nu}^{(k)} a_{-\nu}
    \exp\left(\lambda_4 \sum_{k=1}^{\infty}\frac{a_{-k}}{k} z^k\right)
    \ket{\lambda_1+\lambda_2+\lambda_4+n\lambda_+}\, w^k.
\end{align*}
Each coefficient $r_{\nu}^{(k)}$ is a linear combination of integrals of the form
\begin{align*}
    \int_{\gamma'} ds\,
    \prod_{i=1}^{n} s_i^{\lambda_+\lambda_2+j_i}
    (1-zs_i)^{\lambda_+\lambda_1+l_i}
    \prod_{1 \leq i<j \leq n}(s_i-s_j)^{\lambda_+^2},
    \qquad
    \left(\sum_{i=1}^n j_i,\sum_{i=1}^n l_i \leq k\right),
\end{align*}
whose coefficients are independent of $\epsilon$. Since $r_{\emptyset}^{(0)}$ is independent of $z$, it follows that
\begin{equation*}
     \ket{R_k(z)}
    =
    \sum_{|\nu|\leq k}
    \frac{r_{\nu}^{(k)}}{r_{\emptyset}^{(0)}} a_{-\nu}
    \exp\left(\lambda_4 \sum_{k=1}^{\infty}\frac{a_{-k}}{k} z^k\right)
    \ket{\lambda_1+\lambda_2+\lambda_4+n\lambda_+}.
\end{equation*}
By \eqref{eq limit rvo to ivo 1} and \eqref{eq limit rvo to ivo 3}, the ratio
$r_{\nu}^{(k)}/r_{\emptyset}^{(0)}$ converges, and
\begin{equation*}
    \exp\left(\lambda_4 \sum_{k=1}^{\infty}\frac{a_{-k}}{k} z^k\right)
    \ket{\lambda_1+\lambda_2+\lambda_4+n\lambda_+}
    \to
    e^{c_1 a_{-1}}\ket{c_0+\lambda_2+n\lambda_+}
    \qquad (\epsilon\to 0).
\end{equation*}
Hence each $ \ket{R_k(z)}$ converges.

Finally, the prefactor degenerates as
\begin{align*}
    (-1)^A z^{-\Delta_5+\Delta_4+\Delta_3+A} C_{0,1}
    &=
    w^{\beta(c_0-2\rho+\beta)-2\Delta(\lambda_2)}
    \left(1-\frac{z}{w}\right)^A \\
    &\to
    w^{\beta(c_0-2\rho+\beta)-2\Delta(\lambda_2)} e^{-c_1\beta/w}.
\end{align*}
This completes the proof.
\end{proof}

The above theorem implies the following corollary.
\begin{corollary}\label{cor RCB ICB}
    The 4-point regular conformal block
    \begin{align*}
        (-1)^A z^{\Delta_3+\Delta_4-\Delta_5+A}\braket{\Delta_5|\Phi_{\Delta_5, \Delta_3} ^{\Delta_4} (z) \Phi_{\Delta_3, \Delta_1} ^{\Delta_2} (t) |\Delta_1}
    \end{align*}
    degenerates to the 3-point irregular conformal block
    \begin{align*}
        \bra{\Delta_5 }\cdot \Phi_{\Lambda', \Lambda}^{\Delta}(t) \ket{\Lambda}
        =&
        t^{\beta(c_0  -2 \rho+\beta) - 2 \Delta} e^{-c_1 \beta/t} \left(1+O(t)\right)
    \end{align*}
    by the limit \eqref{eq limit rvo to ivo 1}--\eqref{eq limit rvo to ivo 3}.
    \end{corollary}

\begin{proof}
Since \begin{align*}
         \bra{\Delta_5}\cdot \left(z^{\Delta_3+\Delta_4-\Delta_5}\Phi_{\Delta_5, \Delta_3}^{\Delta_4}(z)\ket{\Delta_3}\right)=1,    \quad
         z^{\Delta_3+\Delta_4-\Delta_5}\Phi_{\Delta_5, \Delta_3}^{\Delta_4}(z)\ket{\Delta_3}\to \ket{\Lambda'},
         \end{align*}
we obtain the desired result by the above theorem.
\end{proof}

Similarly, we obtain the following dual version. 
We use the limit obtained from \eqref{eq limit rvo to ivo 1} and 
\eqref{eq limit rvo to ivo 2} by replacing
    $(\lambda_1,\lambda_3,\lambda_4,\beta,\Delta)$ with $(2\rho-\lambda_5,\,2\rho-\lambda_3,\,\lambda_2,\,-\beta,\,\Delta_4)$, and by replacing $z=\epsilon$ with $z=1/\epsilon$. Explicitly,
\begin{align}
    \lambda_2
    &=
    \frac{c_1}{\epsilon}
    +\frac{c_0}{2},
    \quad
    \lambda_5
    =
    \frac{c_1}{\epsilon}
    -\frac{c_0}{2}
    +2\rho,
    \quad
    \lambda_3
    =
    \frac{c_1}{\epsilon}
    -\frac{c_0}{2}
    +2\rho+\beta,
    \label{eq dual limit rvo to ivo 1}
    \\
    A
    &=
    -\frac{c_1\beta}{\epsilon}
    -\frac{1}{2}\beta(c_0-2\rho-\beta)
    -\Delta_4,
    \label{eq dual limit rvo to ivo 2}
    \\
    z
    &=
    \frac1\epsilon,
    \quad
    \epsilon\to0.
    \label{eq dual limit rvo to ivo 3}
\end{align}
Here $\Delta_i=\Delta(\lambda_i)$ for $i=2,3,5$, and $c_0$ and $c_1$ are the parameters of the limiting rank-one irregular vector in the parametrization given by \eqref{eq Lambda n}.

\begin{corollary}\label{cor dual RCB ICB}
The 4-point dual regular conformal block
\begin{align*}
    (-1)^{-A}z^{-\Delta_3+\Delta_2+\Delta_1-A}
    \braket{
    \Delta_5|
    \Phi_{\Delta_5,\Delta_3}^{\Delta_4,*}(t)
    \Phi_{\Delta_3,\Delta_1}^{\Delta_2,*}(z)
    |\Delta_1}
\end{align*}
degenerates to the 3-point dual irregular conformal block
\begin{align*}
    \braket{
    \Lambda|
    \Phi_{\Lambda,\Lambda'}^{\Delta_4,*}(t)
    |\Delta_1}
\end{align*}
by the limit \eqref{eq dual limit rvo to ivo 1}--\eqref{eq dual limit rvo to ivo 3}.
Here \(\Lambda\) is parametrized by \eqref{eq Lambda n} with \(r=1\) and parameters \(c_0,c_1\), and \(\Lambda'\) is given by
\begin{equation*}
    \Lambda_1'=\Lambda_1-c_1\beta,
    \qquad
    \Lambda_2'=\Lambda_2.
\end{equation*}
\end{corollary}

\subsection{Irregular vertex operator} \label{subsec ivo degeneration}

Consider the composition of two irregular vertex operators
\[
\Phi_{\Lambda',\tilde{\Lambda}}^{\Delta_z}(z)\Phi_{\tilde{\Lambda},\Lambda}^{\Delta_w}(w)
: M_{\Lambda}^{[r]} \to M_{\Lambda'}^{[r]}.
\]
We introduce the following rearranged expansion:
\begin{equation}\label{eq rearranged expansion of IVO}
    \Phi_{\Lambda',\tilde{\Lambda}}^{\Delta_z}(z)\Phi_{\tilde{\Lambda},\Lambda}^{\Delta_w}(w)\ket{\Lambda}
    =
    z^{\alpha_z} w^{\alpha_w}
    \exp\left(\sum_{j=1}^{r}\left(\frac{\beta_j^{(z)}}{z^j}+\frac{\beta_j^{(w)}}{w^j}\right)\right)
    \left(1-\frac{w}{z}\right)^A
    \sum_{k=0}^{\infty} \ket{R_k(z)}\, w^k,
\end{equation}
where $ \ket{R_k(z)}\in M_{\Lambda'}^{[r]}$, $\Lambda_n$ is parameterized as in Proposition~2.2, and
\begin{equation}\label{eq IVO composition parameterization}
    \beta_r^{(w)}=-\frac{\lambda_r\beta}{r},\qquad
    \beta_r^{(z)}=-\frac{\lambda_r\lambda_z}{r},\qquad
    \Delta_z=\Delta(\lambda_z),\qquad
    \Delta_w=\Delta(\lambda_w).
\end{equation}
The remaining parameters are given by \eqref{eq alpha} and \eqref{eq beta}.

The precise meaning of \eqref{eq rearranged expansion of IVO} is as follows.
After the factor
\begin{equation*}
    z^{\alpha_z} w^{\alpha_w}
    \exp\left(\sum_{j=1}^{r}\left(\frac{\beta_j^{(z)}}{z^j}+\frac{\beta_j^{(w)}}{w^j}\right)\right),
\end{equation*}
 is removed, the remaining equality is understood as an identity in $M_{\Lambda'}^{[r]}((z))[[w]]$.
The factor $\left(1-w/z\right)^A$ is expanded by the formal binomial series above. The coefficients $\ket{R_k(z)}\in M_{\Lambda'}^{[r]}((z))$ are then uniquely defined by equating the coefficients of \(w^k\). Thus, \eqref{eq rearranged expansion of IVO} is the irregular analogue of the rearranged expansion \eqref{eq rank 0 to 1 re ex}.

Since the parameter $\beta_r^{(z)}$ depends on $\lambda_r$ and $\lambda_z$, the vertex operator $\Phi_{\Lambda',\tilde{\Lambda}}^{\Delta_z}(z)$ is special. In fact, in what follows we take it to be the free field vertex operator $:e^{\lambda_z\varphi(z)}:$. The reason is that, at present, we do not know how to take the limit of $\Phi_{\Lambda',\Lambda}^{\Delta}(z)\ket{\Lambda}$ as $z\to 0$ in full generality, whereas this limit can be computed explicitly in the free field realization.

Indeed, we have
\begin{align*}
    :e^{\lambda_z \varphi(z)}:
    e^{\sum_{k=1}^{r}\frac{\lambda_k}{k}a_{-k}}
    \ket{\lambda_0}
    &=
    z^{\lambda_z\lambda_0}
    e^{-\lambda_z\sum_{k=1}^{r}\frac{\lambda_k}{k z^k}}
    e^{\lambda_z\sum_{k=r+2}^{\infty}\frac{a_{-k}}{k}z^k}
    e^{\sum_{k=1}^{r+1}\frac{\lambda_k+\lambda_z z^k}{k}a_{-k}}
    \ket{\lambda_0+\lambda_z},
\end{align*}
where $\lambda_{r+1}=0$. If we set
\begin{equation}\label{eq limit irregular vector}
    \lambda_k+\lambda_z z^k=c_k \quad (k=1,2,\ldots,r+1),\qquad
    \lambda_0+\lambda_z=c_0,
\end{equation}
then, in the limit $z\to 0$, we obtain the irregular vector
\begin{equation}\label{eq free field parameter irregular r+1}
    \exp\left(\sum_{k=1}^{r+1}\frac{c_k}{k}a_{-k}\right)\ket{c_0}.
\end{equation}
Here $c_0,\ldots,c_{r+1}$ are the parameters of the limiting rank
$r+1$ irregular vector in the parametrization given by \eqref{eq Lambda n}.
Furthermore, by using $:e^{\lambda_z\varphi(z)}:$, we obtain a general irregular vertex operator between irregular Verma modules of rank $r+1$ as a limit.

\begin{example}
When $r=1$, the coefficient $ \ket{R_1(z)}$ is given by
\begin{align*}
    \ket{R_1(z)}
    =&\, \frac{1}{z}\Bigl(A + (\Lambda_1-\beta_1^{(w)}) d_1 + z(d_{\emptyset}+d_1\Delta_z)\Bigr) \ket{R_0(z)}
    - \frac{d_1}{z}L_1 \ket{R_0(z)}
    + d_1 L_0 \ket{R_0(z)},
\end{align*}
where 
\begin{equation*}
    d_1 = - \frac{\beta_1^{(w)}}{2\Lambda_2}, \qquad
    d_{\emptyset}
    = - \frac{-4 \Delta_w (\Lambda_1 - \beta_1^{(w)}) \Lambda_2 - (\Lambda_1-\beta_1^{(w)})^2 \beta_1^{(w)} + (\Lambda_1-\beta_1^{(w)})(\beta_1^{(w)})^2}{8 \Lambda_2^2}.
\end{equation*}
In view of \eqref{eq limit irregular vector}, this expression admits a finite limit provided that $A$ is chosen appropriately.
\end{example}

\begin{theorem}\label{Thm IVO to IVO}
The coefficients $\ket{R_k(z)}$ converge coefficientwise to the coefficients of 
\[\Phi_{\Gamma',\Gamma}^{\Delta_w}(w)\ket{\Gamma},\]
in the limit
\begin{align}
    \lambda_z &= \frac{c_{r+1}}{\epsilon^{r+1}} + \frac{c_r}{\epsilon^r} + \cdots + \frac{c_1}{\epsilon} + \frac{c_0}{2}, \label{eq limit r 1} \\
    \lambda_j &= -\frac{c_{r+1}}{\epsilon^{r-j+1}} - \frac{c_r}{\epsilon^{r-j}} - \cdots - \frac{c_{j+1}}{\epsilon} + \frac{c_0}{2}\delta_{j,0}
    \qquad (j=0,1,\ldots,r), \label{eq limit r 2}\\
    A &= \frac{c_{r+1}\beta}{\epsilon^{r+1}} + \cdots + \frac{c_1\beta}{\epsilon}
    + \frac{\beta}{2}(c_0-2\rho+\beta)-\Delta_w,
    \qquad z=\epsilon,\qquad \epsilon\to 0, \label{eq limit r 3}
\end{align}
where
\begin{equation}\label{eq Gamma n}
    \Gamma_n = \frac{1}{2}\sum_{k=n-r-1}^{r+1} c_{n-k}c_k
    - \delta_{n,r+1}(r+2)\rho c_{r+1}
    \quad (n=r+1,r+2,\ldots,2r+2),
\end{equation}
and 
\begin{equation*}
    \Gamma'_{r+1} = \Gamma_{r+1}+c_{r+1}\beta, \quad \Gamma_{n+r} ' = \Gamma_{n+r} \quad (n =2, \ldots r+2).
\end{equation*}
\end{theorem}

\begin{lemma}\label{lem recursive limit ivo to ivo}
The coefficients $ \ket{R_k(z)}$ $(k\geq 0)$ of the rearranged expansion satisfy the recursive relations
\begin{align*}
    L_n  \ket{R_k(z)}
    =&\,
    \left(
    \delta_{n,r}\Lambda_r + \cdots + \delta_{n,2r}\Lambda_{2r}
    + z^n\left(\alpha_z - \sum_{i=1}^{r}\frac{i\beta_i^{(z)}}{z^i} + (n+1)\Delta_z + z\partial_z\right)
    \right) \ket{R_k(z)} \\
    &+
    \sum_{i=1}^{n-1} A z^i\ket{R_{k-n+i}(z)} - \sum_{i=1}^{r} i\beta_i^{(w)}\ket{R_{k-n+i}(z)}
    \\
    &+
    \left(
    A+\alpha_w+(n+1)\Delta_w+k-n
    \right)\ket{R_{k-n}(z)}
    \qquad (n\geq r).
\end{align*}
Assume that all $ \ket{R_k(z)}$ converge to $\ket{u_k}$. Then, by \eqref{eq alpha} and \eqref{eq beta}, these relations converge, under the limit of Theorem \ref{Thm IVO to IVO}, to
\begin{align*}
    L_n \ket{u_k}
    =&\,
    \sum_{i=1}^{r+2}\delta_{n,r+i}\Gamma_{r+i}\ket{u_k}
    + \beta\sum_{i=1}^{r+1} c_i \ket{u_{k-n+i}} \\
    &+
    \bigl(\alpha+(n+1)\Delta_w+k-n\bigr)\ket{u_{k-n}}
    \qquad (n\geq r+1),
\end{align*}
where $\alpha$ is given by the rank $r+1$ version of \eqref{eq alpha}. These are precisely the defining relations \eqref{eq recursive relation for IVO} for the rank $r+1$ irregular vertex operator. Consequently, if the coefficients $ \ket{R_k(z)}$ converge, then their limits coincide with the coefficients of $\Phi_{\Gamma',\Gamma}^{\Delta_w}(w)$.
\end{lemma}
\begin{proof}
This is immediate from \eqref{eq alpha} and \eqref{eq beta}.
\end{proof}

\begin{proof}[Proof of Theorem \ref{Thm IVO to IVO}]
By Lemma \ref{lem recursive limit ivo to ivo}, it remains to prove the convergence of the coefficients $ \ket{R_k(z)}$ under the scaling limit \eqref{eq limit r 1}--\eqref{eq limit r 3}. As in the proof of Theorem \ref{Thm RVO to IVO}, it is sufficient to establish this for infinitely many irregular vertex operators arising from the free field realization.

\medskip
\noindent\textbf{Step 1. Free field representation.}
We first consider the free field realization of the composition:
\begin{align*}
    &:e^{\lambda_z \varphi(z)}::e^{\lambda_w \varphi(w)}:
    Q_{+}^n \exp\left( \sum_{k=1}^{r} \frac{\lambda_k}{k} a_{-k} \right) \ket{\lambda_0} \\
    =&\,
    z^{\lambda_z \lambda_0} w^{\lambda_w \lambda_0} (z-w)^{\lambda_z \lambda_w}
    \exp\left( \sum_{k=1}^{r} \frac{-\lambda_k \lambda_z}{k z^k} \right)
    \exp\left( \sum_{k=1}^{r} \frac{-\lambda_k \lambda_w}{k w^k} \right) \\
    &\times \int_{\Delta} dt\,
    \prod_{1 \leq i<j \leq n} (t_i - t_j)^{\lambda_+ ^2}
    \prod_{i=1}^{n}\left[ t_i^{\lambda_+ \lambda_0} (z-t_i)^{\lambda_+ \lambda_z} (w-t_i)^{\lambda_+ \lambda_w}
    \exp\left( \sum_{k=1}^{r} \frac{-\lambda_k \lambda_+}{k t_i^k} \right)\right] \\
    &\times
    \exp\left( \lambda_{+} \sum_{i=1}^n\sum_{k=1}^{\infty} \frac{a_{-k}}{k}t_i^k \right)
    \exp\left( \lambda_w \sum_{k=1}^{\infty} \frac{a_{-k}}{k}w^k \right)
    \exp\left( \lambda_z \sum_{k=1}^{\infty} \frac{a_{-k}}{k}z^k \right) \\
    &\times
    \exp\left( \sum_{k=1}^{r} \frac{\lambda_k}{k} a_{-k} \right)
    \ket{\lambda_0+\lambda_z + \lambda_w + n\lambda_+}.
\end{align*}
 To compute the coefficients in this free field representation, let $z$ and $w$ be positive real numbers such that $z>w>0$. We assume that the parameters $\lambda_k$ $(k=0,1,\ldots,r,+,z,w)$ are also positive, and define
\begin{equation*}
    \Delta=\{ (t_1,\ldots,t_n)\in\mathbb{R}^n \mid w>t_1>t_2>\cdots>t_n>0 \}.
\end{equation*}
Since $\Lambda'_r=\Lambda_r+\lambda_r(\lambda_w+n\lambda_+)$ in the above realization, the parameter $\beta$ in \eqref{eq limit r 3} is equal to $\lambda_w+n\lambda_+$.

\medskip
\noindent\textbf{Step 2. Change of variables.}
In order to obtain the expansion at $w=0$, we introduce the change of variables
\begin{align*}
    t_i = \frac{w(1 - w^r z s_i)}{1 - w^{r+1}s_i}.
\end{align*}
Then the above expression takes the form
\begin{equation}\label{eq composition r 2}
\begin{split}
     &(-1)^{\frac{n(n-1)}{2}\lambda_+^2 +n} C_{r,r+1}\\
     &\times z^{\frac{n(n-1)}{2}\lambda_+^2 + n \lambda_+ \lambda_w + n}
     \int_{\Delta'} ds\,
     \prod_{1 \leq i<j \leq n} (s_i - s_j)^{\lambda_+ ^2}
     \prod_{i=1}^{n} \left[s_i^{\lambda_+ \lambda_w}
     e^{-\lambda_r \lambda_+ z s_i}
     (1-w^{r+1} s_i)^{-\lambda_+ \lambda_0 -\lambda_+ \lambda_z -\lambda_+ \lambda_w -(n-1)\lambda_+^2-2}\right. \\
    &\times \left.(1 - w^r zs_i)^{\lambda_+ \lambda_0}
    \exp\left( \sum_{k=1}^{r} -\frac{\lambda_k \lambda_+}{k}
    \left(
    \left( \frac{1-w^{r+1}s_i}{w(1 - w^rz s_i)} \right)^k - \left( \frac{1}{w} \right)^k - r zs_i \delta_{k,r}
    \right) \right) \right]\\
    &\times
    \exp\left( \lambda_{+} \sum_{i=1}^n\sum_{k=1}^{\infty} \frac{a_{-k}}{k}\left( \frac{w(1 - w^r z s_i)}{1- w^{r+1} s_i} \right)^k \right)
    \exp\left( \lambda_w \sum_{k=1}^{\infty} \frac{a_{-k}}{k}w^k \right) \\
    &\times
    \exp\left( \lambda_z \sum_{k=1}^{\infty} \frac{a_{-k}}{k}z^k \right)
    \exp\left( \sum_{k=1}^{r} \frac{\lambda_k}{k} a_{-k} \right)
    \ket{\lambda_0+\lambda_z + \lambda_w + n\lambda_+},
\end{split}
\end{equation}
where
\begin{align*}
    C_{r,r+1}
    =
    &\, z^{(\lambda_0 + \lambda_w + n\lambda_+)\lambda_z}
    \exp\left( \sum_{k=1}^{r} \frac{-\lambda_z \lambda_k}{k z^k} \right) \\
    &\times
    w^{\lambda_w \lambda_0 + n\lambda_+ \lambda_0  + (r+1) n\lambda_+ \lambda_w + (r+1)\frac{n(n-1)}{2} \lambda_+^2+ (r+1)n}
    \exp\left( \sum_{k=1}^{r} \frac{-\lambda_k (\lambda_w + n\lambda_+)}{k w^k} \right) \\
    &\times
    \left(1-\frac{w}{z}\right)^{(\lambda_w +n\lambda_+ )\lambda_z +n\lambda_+ \lambda_w +\frac{n(n-1)}{2} \lambda_+^2 + n},
\end{align*}
and
\begin{equation*}
    \Delta'=\{ (s_1,\ldots,s_n)\in\mathbb{R}^n \mid 0<s_1<s_2<\cdots<s_n<1/(w^r z) \}.
\end{equation*}
We note that the exponent of \(1-w/z\) is precisely \(A\) in \eqref{eq limit r 3}.

\medskip
\noindent\textbf{Step 3. Expansion in powers of \(w\).}
We next examine the asymptotic expansion of the integral in \eqref{eq composition r 2}. Let $M$ be a non-negative integer, and consider an integral of the form
\begin{align}\label{eq int rep IVO}
    \int_{\Delta'} ds\, \phi(s) \prod_{i=1}^n \left[ F_i^M(w,s_i) (1 - w^r zs_i)^{c} e^{-X^M(s_i)}\right],
\end{align}
where
\begin{align*}
    \phi(s) &= \prod_{1 \leq i<j \leq n} (s_i - s_j)^{\lambda_+ ^2} \prod_{i=1}^{n} s_i^{a}e^{-\lambda_r \lambda_+z s_i }, \\
    F_i^M(w,s_i)&= (1-w^{r+1} s_i)^{b}e^{-f^M(w,s_i)}, \\
    f^M(w,s_i)&=\sum_{k=1}^{r}f_k^M(w,s_i) -\lambda_r \lambda_+ zs_i, \\
    f_k^M(w,s_i) &= \frac{\lambda_k \lambda_+}{k}
    \left(
    \left( \frac{1-w^{r+1}s_i}{w} \sum_{m=0}^{M}(w^rzs_i)^m\right)^k -\frac{1}{w^k}
    \right), \\
    X^M(s_i)&=\sum_{k=1}^{r}X_k^M(s_i), \\
    X_k^{M}(s_i) &= \frac{\lambda_k \lambda_+}{k}\left( \frac{1-w^{r+1}s_i}{w}\right)^k
    \sum_{l=1}^{k}
      \binom{k}{l}
    \left( \sum_{m=0}^{M}(w^rzs_i)^m \right)^{k-l}
    \left( \frac{(w^rzs_i)^{M+1}}{1-w^rzs_i} \right)^l,
\end{align*}
with $a,c>0$ and $b\in\mathbb{R}$. We note that $f^M(w,s_i)$ is a polynomial in $w$ with vanishing constant term.

By Lemma \ref{lem asym IVO}, the integral \eqref{eq int rep IVO} admits an asymptotic expansion as $w\to 0$. It follows that \eqref{eq composition r 2} can be expanded as
\begin{align*}
    (-1)^{\frac{n(n-1)}{2}\lambda_+^2 +n} C_{r,r+1}
    \sum_{m=0}^{\infty}w^m \sum_{|\nu|\leq m} r_{\nu} ^{(m)}(z)a_{-\nu}
    \exp\left( \lambda_z \sum_{k=1}^{\infty} \frac{a_{-k}}{k}z^k \right)\\
    \times
    \exp\left( \sum_{k=1}^{r} \frac{\lambda_k}{k} a_{-k} \right)
    \ket{\lambda_0+\lambda_z + \lambda_w + n\lambda_+}.
\end{align*}
Each coefficient $r_{\nu} ^{(m)}(z)$ is a linear combination of integrals of the form
\begin{align*}
    z^{\frac{n(n-1)}{2}\lambda_+^2 + n \lambda_+ \lambda_w + n}
    \int_{\Delta^{''}} ds\,
    \prod_{1 \leq i<j \leq n} (s_i - s_j)^{\lambda_+ ^2}
    \prod_{i=1}^{n} s_i^{\lambda_+ \lambda_w+k_i} e^{-\lambda_r \lambda_+ z s_i }
    \qquad (k_i \leq m),
\end{align*}
where
\begin{equation*}
    \Delta^{''}= \{ (s_1,\ldots,s_n)\in\mathbb{R}^n \mid 0<s_1<s_2<\cdots<s_n<\infty \}.
    \end{equation*}
In particular,
\begin{align*}
    r_{\emptyset} ^{(0)}
    =&\,
    z^{\frac{n(n-1)}{2}\lambda_+^2 + n \lambda_+ \lambda_w + n}
    \int_{\Delta^{''}} ds\,
    \prod_{1 \leq i<j \leq n} (s_i - s_j)^{\lambda_+ ^2}
    \prod_{i=1}^{n} s_i^{\lambda_+ \lambda_w} e^{- \lambda_r \lambda_+z s_i }.
\end{align*}
Hence
\begin{align*}
     \ket{R_k(z)}
    =
    \sum_{|\nu|\leq k}\frac{r_{\nu} ^{(k)}}{r_{\emptyset} ^{(0)}}a_{-\nu}
    \exp\left( \lambda_z \sum_{m=1}^{\infty} \frac{a_{-m}}{m}z^m \right)
    \exp\left( \sum_{m=1}^{r} \frac{\lambda_m}{m} a_{-m} \right)
    \ket{\lambda_0+\lambda_z + \lambda_w + n\lambda_+}.
\end{align*}
Since all $r_\nu^{(k)}$ contain the common factor
\[
z^{\frac{n(n-1)}{2}\lambda_+^2 + n \lambda_+ \lambda_w + n},
\]
this factor cancels in the ratio $r_\nu ^{(k)}/r_\emptyset^{(0)}$.

\medskip
\noindent\textbf{Step 4. Passage to the limit.}
It remains to verify that each factor in the above expression admits the limit stated in the theorem.

First, the prefactor in \eqref{eq composition r 2} satisfies
\begin{align*}
    &\epsilon^{-\alpha_z+A}\exp\left(-\sum_{j=1}^{r}\frac{\beta_j^{(z)}}{\epsilon^j} \right)C_{r,r+1} \\
    &=
    w^{c_0 \beta + \frac{r+2}{2}\beta(\beta-2\rho) - (r+2)\Delta(\lambda_w)}
    \exp\left( -\sum_{k=1} ^{r} \frac{\lambda_k\beta}{k w^k} \right)
    \left(1-\frac{z}{w}\right)^A \\
    &\to
    w^{c_0 \beta + \frac{r+2}{2}\beta(\beta-2\rho) - (r+2)\Delta(\lambda_w)}
    \exp\left(-\sum_{k=1}^{r+1}\frac{c_k\beta}{kw^k}\right).
\end{align*}

Next, consider the factor
\begin{align*}
    (1 - w^r z s )^{\lambda_+ \lambda_0}
    \exp\left( - \sum_{k=1} ^{r} \frac{\lambda_k \lambda_+}{k}
    \left( \left( \frac{1-w^{r+1}s}{w(1 - w^r z s)}  \right)^k - \left( \frac{1}{w} \right)^k  \right) \right).
\end{align*}
Its exponent is equal to
\begin{align*}
    &\lambda_+ \lambda_0 \log(1 - w^r z s)
    - \sum_{k=1} ^{r} \frac{\lambda_k \lambda_+}{k w^k}
    \left( \left( \frac{1-w^{r+1}s}{1 - w^r z s}  \right)^k - 1  \right) \\
    =&
    -\lambda_+ \lambda_0 \sum_{N=1} ^{\infty}\frac{z^N s^N}{N}w^{rN}
    -\sum_{N=1}^{\infty} \sum_{k=1} ^{r} \frac{\lambda_k \lambda_+}{k w^k}
    \sum_{i+j=N,\, 0\leq i \leq k} (-1)^i
      \binom{k}{l}
        \binom{k+j-1}{j}
    z^j s^N w^{rN+i} \\
    =&
    -\sum_{N=1} ^{\infty} \sum_{k=0} ^{r} w^{r N - k} s^N
    \sum_{i=0} ^{\min{N,r-k}} \lambda_{k+i} \lambda_{+} (-1)^i
    \frac{(k+N-1)!}{i!k!(N-i)!} z^{N-i}.
\end{align*}
Since $\lambda_{k+i} z^{N-i}$ converges whenever $N \geq r-k+1$, it is enough to consider the case $N\leq r-k$. In this range, the terms with negative powers of $\epsilon$ in $\lambda_{k+i}\epsilon^{N-i}$ coincide with those of $\lambda_{k+N}$:
\begin{align*}
    \lambda_{k+i} \epsilon^{N-i}
    = \lambda_{k+N} -\frac{c_0}{2} \delta_{k+N,0}
    - \sum_{j=0}^{N-i-1} c_{k+N-j}\epsilon^{j}
    + \frac{c_0}{2} \delta_{k+i,0} \epsilon^{N-i}.
\end{align*}
Accordingly, the terms with negative powers of $\epsilon$ cancel, since
\begin{align*}
    \sum_{i=0} ^{N} (-1)^i \frac{(k+N-1)!}{i!k!(N-i)!}=0
\end{align*}
for $N\leq r-k$. Therefore this factor admits the required limit.

Finally, the Fock-space vector converges to the irregular vector of rank $r+1$:
\begin{align*}
    &\exp\left( \lambda_z \sum_{k=1} ^{\infty} \frac{a_{-k}}{k}z^k \right)
    \exp\left( \sum_{k=1} ^{r} \frac{\lambda_k}{k} a_{-k} \right)
    \ket{\lambda_0+\lambda_z + \lambda_w + n\lambda_+} \\
    =&\,
    \exp\left(
    \sum_{k=1} ^{r} \frac{\lambda_k + \lambda_z z^k}{k}a_{-k}
    + \lambda_z\sum_{k=r+1} ^{\infty} \frac{a_{-k}}{k}z^k
    \right)
    \ket{\lambda_0+\lambda_z + \lambda_w + n\lambda_+}.
\end{align*}
By \eqref{eq limit r 1} and \eqref{eq limit r 2}, this converges to the rank $r+1$ irregular vector $\ket{\Gamma'}$ determined by the Fock parameters $(c_0 + \beta,\ldots,c_{r+1})$. Hence each coefficient $ \ket{R_k(z)}$ converges. The proof is complete.
\end{proof}

Using Corollary \ref{cor RCB ICB}, the four-point regular conformal block
\begin{align*}
    z_1^{-\Delta_5 + \Delta_4 + \Delta_3+A}
    \braket{0|\Phi_{0,\Delta_5}^{\Delta_5}(z_2)\Phi_{\Delta_5, \Delta_3}^{\Delta_4}(z_1)\Phi_{\Delta_3,\Delta_1}^{\Delta_2}(w)|\Delta_1}
\end{align*}
degenerates, under the limit of Theorem \ref{Thm RVO to IVO} with $z_1\to 0$, to the three-point irregular conformal block
\begin{align*}
    \braket{0|\Phi_{0,\Delta_5}^{\Delta_5}(z_2)\Phi_{\Lambda', \Lambda}^{\Delta_w}(w)|\Lambda}.
\end{align*}
We note that the irregular vector $\ket{\Lambda'}$ belongs to the completion of the Verma module $M_{\Delta_5}$. Applying Theorem \ref{Thm IVO to IVO} with $r=1$ to this three-point block, we obtain the following corollary.

\begin{corollary}\label{cor limit of ICB 2}
The three-point irregular conformal block
\begin{align*}
    (-1)^{-A} z_2^{2 \Delta_5+A}e^{-\frac{\Lambda_1'}{z_2}}
    \braket{0|\Phi_{0,\Delta_5}^{\Delta_5}(z_2)\Phi_{\Lambda', \Lambda}^{\Delta_w}(w)|\Lambda}
\end{align*}
degenerates, under the limit of Theorem \ref{Thm IVO to IVO}, to the two-point irregular conformal block
\begin{align*}
    \braket{0|\cdot\, \Phi_{\Gamma', \Gamma}^{\Delta_w}(w)|\Gamma}.
\end{align*}
\end{corollary}

\begin{proof}
Using the rearranged expansion for the composition  $\Phi_{0,\Delta_5}^{\Delta_5}(z_2)\Phi_{\Lambda', \Lambda}^{\Delta_w}(w)\ket{\Lambda}$, we obtain
\begin{align*}
    &(-1)^{-A}z_2^{2 \Delta_5+A}e^{-\frac{\Lambda_1'}{z_2}}
    \braket{0|\Phi_{0,\Delta_5}^{\Delta_5}(z_2)\Phi_{\Lambda', \Lambda}^{\Delta_w}(w)|\Lambda}\\
    =&\,
    (-1)^{-A} z_2^{2 \Delta_5+A}e^{-\frac{\Lambda_1'}{z_2}}
    w^{\alpha} e^{\beta/w}
    \left( 1-\frac{w}{z_2} \right)^A
    \sum_{ k=0}^{\infty}w^k\sum_{\nu \in \mathbb{Y}}b_{\nu}^{(k)}
    \braket{0|L_{-\nu}\Phi_{0, \Delta_5}^{\Delta_5}(z_2)|\Lambda'}\\
    =&\,
    z_2^{2 \Delta_5}e^{-\frac{\Lambda_1'}{z_2}}w^{\alpha+A} e^{\beta/w}
    \left( 1-\frac{z_2}{w} \right)^A
    \sum_{ k=0}^{\infty}b_{\emptyset}^{(k)}
    \braket{0|\Phi_{0, \Delta_5}^{\Delta_5}(z_2)|\Lambda'}w^k,
\end{align*}
where $L_{-\nu}=L_{-\nu_1+2}\cdots L_{-\nu_\ell+2}$ for a partition $\nu=(\nu_1\geq\cdots\geq\nu_\ell)$. Since
\begin{align*}
    \braket{0|\Phi_{0,\Delta_5}^{\Delta_5}(z_2)|\Lambda'}=z_2 ^{-2 \Delta_5 } e^{\frac{\Lambda_1'}{z_2}},
\end{align*}
and Theorem \ref{Thm IVO to IVO} implies that each $b_{\emptyset}^{(k)}$ degenerates to the constant term $\tilde{b}_{\emptyset}^{(k)}$ of $\Phi_{\Gamma', \Gamma}^{\Delta_w}(w)\ket{\Gamma}$, where
\begin{equation*}
    \Phi_{\Gamma', \Gamma}^{\Delta_w}(w)\ket{\Gamma}
    =w^{\tilde{\alpha}}e^{\sum_{i=1}^{r+1}\frac{\tilde{\beta}_i}{w^i}}
    \sum_{k=0}^{\infty} w^k\sum_{\nu\in\mathbb{Y}}\tilde{b}_{\nu}^{(k)}L_{-\nu}\ket{\Gamma'},
\end{equation*}
the assertion follows.
\end{proof}

Similarly, we obtain the following dual version.

\begin{corollary}\label{cor dual limit of ICB 2}
The three-point dual irregular conformal block
\begin{equation*}
    (-1)^A z_2^A e^{-\Lambda_1' z_2}
    \braket{\Lambda|
    \Phi_{\Lambda,\Lambda'}^{\Delta_w,*}(w)
    \Phi_{\Delta_5,0}^{\Delta_5,*}(z_2)
    |0}
\end{equation*}
degenerates to the two-point dual irregular conformal block
\begin{equation*}
    \braket{
    \Gamma|
    \Phi_{\Gamma,\Gamma'}^{\Delta_w,*}(w)
    |0}
\end{equation*}
in the limit \eqref{eq limit r 1}, \eqref{eq limit r 2} with $r=1$, together with
\begin{equation}
    A
    =
    \frac{c_2\beta}{\epsilon^2}
    +\frac{c_1\beta}{\epsilon}
    +\frac{1}{2}\beta(c_0-2\rho-\beta)
    +\Delta_w,
    \quad
    z_2=\frac1\epsilon,
    \quad
    \epsilon\to0.
    \label{eq dual limit r 3}
\end{equation}
Here $\Delta_5=\Delta(\lambda_z)$, $\Lambda$ is parametrized by \eqref{eq Lambda n} with $r=1$ and
$\lambda_0,\lambda_1$, and
\begin{equation*}
    \Lambda_1'=\Lambda_1-\lambda_1\beta,
    \quad
    \Lambda_2'=\Lambda_2.
\end{equation*}
Moreover, $\Gamma$ is parametrized by \eqref{eq Gamma n} with $r=1$, and
\begin{equation*}
    \Gamma_2'=\Gamma_2-c_2\beta,
    \quad
    \Gamma_3'=\Gamma_3,
    \quad
    \Gamma_4'=\Gamma_4.
\end{equation*}
\end{corollary}

\section{Degeneration limit of Painlev\'e $\tau$ functions}\label{sec tau functions}
\subsection{Painlev\'e equations and $\tau$ functions}

  The Painlev\'e equations admit a Hamiltonian formulation
\begin{align*}
    \frac{dq}{dt}=\frac{\partial H}{\partial p}, \qquad
    \frac{dp}{dt}=-\frac{\partial H}{\partial q}.
\end{align*}
For $\mathrm{J}=\mathrm{VI},\mathrm{V},\mathrm{IV}$, the Hamiltonians are given by
\begin{align*}
    t(t-1)H_{\mathrm{VI}}
    =&\, q(q-1)(q-t)p
    \left( p - \frac{2 \theta_0}{q} - \frac{2 \theta_1}{q-1} - \frac{2 \theta_t -1}{q-t} \right)
    \\
    &+ (\theta_0 + \theta_t + \theta_1 + \theta_{\infty})
    (\theta_0 + \theta_t + \theta_1 -\theta_{\infty} -1)q,
    \\
    t H_{\mathrm{V}}
    =&\, (q-1)(pq-2\theta_t)(pq-p+2\theta)-tpq+((\theta+\theta_t)^2-\theta_0^2)q \\
    &+\left(\theta_t-\frac{\theta}{2}\right)t
    -2\left(\theta_t+\frac{\theta}{2}\right)^2,
    \\
    H_{\mathrm{IV}}
    =&\, 2qp^2-(q^2+2tq-\theta_*-\theta_t)p-\theta_t q.
\end{align*}

For a solution $(q(t),p(t))$ of Hamilton's equations, we denote by
\begin{align*}
    H_{\mathrm{J}}(t)=H_{\mathrm{J}}(q(t),p(t);t)
\end{align*}
the corresponding Hamilton function. We then introduce the auxiliary functions
\begin{align*}
    \sigma_{\mathrm{VI}}(t)
    &= t(t-1)H_{\mathrm{VI}}(t)-q(q-1)p
    +(\theta_0+\theta_t+\theta_1+\theta_{\infty})q
    \\
    &\qquad -(\theta_0+\theta_1)^2 t
    +\frac{\theta_1^2+\theta_{\infty}^2-\theta_0^2-\theta_t^2-4\theta_0\theta_t}{2}, \\
    \sigma_{\mathrm{V}}(t)&=tH_{\mathrm{V}}(t), \\
    \sigma_{\mathrm{IV}}(t)&=H_{\mathrm{IV}}(t).
\end{align*}
These functions satisfy the second-order nonlinear differential equations
\begin{align}
&(t-1)^2 t^2 \sigma_{\mathrm{VI}}' (\sigma_{\mathrm{VI}}'')^2 +
    \left(-\sigma_{\mathrm{VI}}'^2
          +2 \left(t \sigma_{\mathrm{VI}}'-\sigma_{\mathrm{VI}}\right) \sigma_{\mathrm{VI}}'(t)
          -\left(\theta_{0}^2-\theta_{1}^2\right)\left(\theta_{t}^2-\theta_{\infty}^2\right)\right)^2
     \label{eq sigma6}\\
  &-\left(\sigma_{\mathrm{VI}}'+(\theta_{0}-\theta_{1})^2\right)
    \left(\sigma_{\mathrm{VI}}'+(\theta_{0}+\theta_{1})^2\right)
    \left(\sigma_{\mathrm{VI}}'+(\theta_{t}-\theta_{\infty})^2\right)
    \left(\sigma_{\mathrm{VI}}'+(\theta_{t}+\theta_{\infty})^2\right)
    = 0, \nonumber\\
&(t\sigma_\mathrm{V}'')^2
      -(\sigma_\mathrm{V}-t\sigma_\mathrm{V}'+2(\sigma_\mathrm{V}')^2)^2
      +\frac{1}{4}\bigl((2\sigma_\mathrm{V}'-\theta)^2-4\theta_0^2\bigr)
       \bigl((2\sigma_\mathrm{V}'+\theta)^2-4\theta_t^2\bigr)
    =0,
    \label{eq sigma5}\\
&\left(\sigma_\mathrm{IV}''\right)^2
      -(t\sigma_\mathrm{IV}'-\sigma_\mathrm{IV})^2
      +4\sigma_\mathrm{IV}'(\sigma_\mathrm{IV}'-\theta_*-\theta_t)
       (\sigma_\mathrm{IV}'-2\theta_t)
    =0.
    \label{eq sigma4}
\end{align}
We denote these equations by  $E_{\mathrm J}(\sigma_{\mathrm J},t)=0$ $(\mathrm{J}=\mathrm{VI},\mathrm{V},\mathrm{IV})$, respectively.

Conversely, if a function $\sigma_{\mathrm{J}}(t)$ $(\mathrm{J}=\mathrm{VI},\mathrm{V},\mathrm{IV})$ satisfies $E_{\mathrm{J}}(\sigma_{\mathrm{J}},t)=0$, then the canonical variables $q(t)$ and $p(t)$ can be recovered as rational functions of $\sigma_{\mathrm{J}}(t)$, $\sigma_{\mathrm{J}}'(t)$, and $\sigma_{\mathrm{J}}''(t)$ \cite{Okamoto,Jimbo}. In particular, the function $q(t)$ satisfies the corresponding Painlev\'e equation.

The $\tau$ functions are defined by
\begin{align}
    \sigma_{\mathrm{VI}} (t) &= t (t-1) \frac{d}{dt}
    \log\left(
    t^{\frac{\theta_0 ^2 + \theta_t^2 - \theta_1^2 - \theta_{\infty}^2}{2}}
    (1-t)^{\frac{\theta_t^2 + \theta_1 ^2 - \theta_0^2 - \theta_{\infty}^2}{2}}
    \tau_{\mathrm{VI}}(t)
    \right),\label{def tau6} \\
    \sigma_{\mathrm{V}}(t) &= t\frac{d}{dt}
    \log\left(
    t^{-\frac{\theta^2}{2}} e^{-\frac{\theta}{2}t}\tau_\mathrm{V}(t)
    \right),\label{def tau5} \\
    \sigma_{\mathrm{IV}}(t) &= \frac{d}{dt}
    \log\left(
    e^{\frac{\theta_*}{2}t^2}\tau_\mathrm{IV}(t)
    \right). \label{def tau4}
\end{align}

\subsection{Expansion of $\tau$ functions at $t=\infty$}

In this subsection, we derive expansions at $t=\infty$ for the $\tau$ functions of $\mathrm{P_V}$ and $\mathrm{P_{IV}}$ by degenerating the $c=1$ conformal-block expansion of the sixth Painlev\'e $\tau$ function.

The $\tau$ function of the sixth Painlev\'e equation admits a Fourier expansion in terms of four-point Virasoro conformal blocks with $c=1$:
\begin{align}
    \tau_{\mathrm{VI}}^{(\infty)} (t)
    &= \sum_{n\in\mathbb{Z}}e^{2\pi i n  \varrho}\,
    C_\mathrm{VI}(\vec{\theta},\sigma+n)
    \braket{\theta_{\infty}^2|
    \Phi_{\theta_{\infty} ^2,\, (\sigma+n)^2} ^{\theta_{t} ^2}(t)
    \Phi_{(\sigma+n)^2,\, \theta_{0}^2}^{\theta_{1}^2}(1)
    |\theta_{0}^2}
    \label{eq tau6 at infty}
\end{align}
at $t=\infty$ \cite{GIL}. Here $ \varrho,\sigma\in\mathbb{C}$, and
\begin{equation*}
	C_\mathrm{VI}(\vec{\theta},\sigma)
    = \frac{\prod_{\epsilon,\epsilon'=\pm}
    G(1+\theta_1+\epsilon\theta_t+\epsilon'\sigma)\,
    G(1+\theta_0+\epsilon\theta_\infty+\epsilon'\sigma)}
    {\prod_{\epsilon=\pm} G(1+2\epsilon\sigma)}.
\end{equation*}

By the Ward identities, the four-point conformal block can be rewritten as
\begin{align*}
    &\braket{\theta_{\infty} ^2|
    \Phi_{\theta_{\infty} ^2,\, (\sigma+n)^2} ^{\theta_{t} ^2}(t)
    \Phi_{(\sigma+n)^2,\, \theta_{0}^2}^{\theta_{1}^2}(1)
    |\theta_{0}^2}
    \\
    &=
    \braket{\theta_{\infty} ^2|
    \Phi_{\theta_{\infty} ^2,\, (\sigma+n)^2} ^{\theta_{t} ^2}\left(\frac{s-z_2}{z_1-z_2}\right)
    \Phi_{(\sigma+n)^2,\, \theta_{0}^2}^{\theta_{1}^2}(1)
    |\theta_{0}^2}
    \\
    &=
    (z_1-z_2)^{\theta_1^2+\theta_0^2+\theta_t^2-\theta_\infty^2}
    \braket{\theta_{\infty} ^2|
    \Phi_{\theta_{\infty} ^2,\, (\sigma+n)^2} ^{\theta_{t} ^2}(s-z_2)
    \Phi_{(\sigma+n)^2,\, \theta_{0}^2}^{\theta_{1}^2}(z_1-z_2)
    |\theta_{0}^2}
    \\
    &=
    (z_1-z_2)^{\theta_1^2+\theta_0^2+\theta_t^2-\theta_\infty^2}
    \braket{\theta_{\infty} ^2|
    \Phi_{\theta_{\infty} ^2,\, (\sigma+n)^2} ^{\theta_{t} ^2}(s)
    \Phi_{(\sigma+n)^2,\, \theta_{0}^2}^{\theta_{1}^2}(z_1)
    \Phi^{\theta_0^2}_{\theta_0^2,0}(z_2)
    |0}.
\end{align*}

We make the corresponding change of variables in the differential equations. For $E_{\mathrm{VI}}(\sigma_{\mathrm{VI}},t)=0$, we set
\begin{equation*}
    t=\frac{s-z_2}{z_1-z_2},
\end{equation*}
whereas for $E_{\mathrm{V}}(\sigma_{\mathrm{V}},t)$, we set
\begin{equation*}
    t=\eta (s-z_2).
\end{equation*}
This yields transformed differential equations $\tilde{E}_{\mathrm{J}}(f(s),s)=0$, where
\[
f(s)=\frac{d}{ds}\log \tau_{\mathrm{J}}(s).
\]

For $E_{\mathrm{IV}}(\sigma_{\mathrm{IV}},t)=0$, we simply put $t=s$. This yields transformed differential equations
\[
    \tilde{E}_{\mathrm{J}}(f(s),s)=0,
    \qquad
    f(s)=\frac{d}{ds}\log \tau_{\mathrm{J}}(s)
    \qquad
    (\mathrm{J}=\mathrm{VI},\mathrm{V},\mathrm{IV}).
\]
We do not write their explicit forms.
These equations are related by the degeneration scheme
\begin{align*}
    \tilde{E}_{\mathrm{VI}}(f(s),s) =0\to
    \tilde{E}_{\mathrm{V}}(f(s),s)=0\to
    \tilde{E}_{\mathrm{IV}}(f(s),s)=0,
\end{align*}
under the dual degeneration limits of conformal blocks stated in Section~3. 
For the degeneration from the four-point dual regular conformal block to the three-point dual irregular conformal block, we use Corollary~\ref{cor dual RCB ICB} with \eqref{eq dual limit rvo to ivo 1}--\eqref{eq dual limit rvo to ivo 3}; for the degeneration from the three-point dual irregular conformal block to the two-point dual irregular conformal block, we use Corollary~\ref{cor dual limit of ICB 2} with \eqref{eq limit r 1}, \eqref{eq limit r 2} specialized to \(r=1\), and \eqref{eq dual limit r 3}. 
In the Painlev\'e parametrization, these limits are written as follows:
\begin{align}
\label{lim 6to5}
&\theta_1 = \frac{\eta}{2 \epsilon} + \frac{\theta}{2}, \qquad
\theta_{\infty} = \frac{\eta}{2 \epsilon} - \frac{\theta}{2}, \qquad
z_1 = \frac{1}{\epsilon}, \qquad \epsilon\to 0,
\\
\label{lim 5to4}
&\theta = -\frac{1}{2\epsilon^2} + \frac{\theta_{*}}{2}, \qquad
\eta = -\frac{1}{\epsilon}, \qquad
\theta_0 = -\frac{1}{2\epsilon^2} -\frac{\theta_{*}}{2}, \qquad
z_2 =\frac{1}{\epsilon}, \qquad \epsilon \to 0.
\end{align}

It follows from \cite{BS} that
\begin{align}
    \tau_{\mathrm{VI}}^{(\infty)} (s,z_1,z_2)
    =&(z_1-z_2)^{\theta_1^2+\theta_0^2+\theta_t^2-\theta_\infty^2}
    \sum_{n\in\mathbb{Z}}e^{2\pi i n  \varrho}
    C_\mathrm{VI}(\vec{\theta},\sigma+n) \nonumber\\
    &\times
    \braket{\theta_{\infty} ^2|
    \Phi_{\theta_{\infty} ^2,\, (\sigma+n)^2} ^{\theta_{t} ^2,*}(s)
    \Phi_{(\sigma+n)^2,\, \theta_{0}^2}^{\theta_{1}^2,*}(z_1)
    \Phi_{\theta_0^2,0}^{\theta_0^2,*}(z_2)
    |0}
    \label{eq tau VI s}
\end{align}
satisfies $\tilde{E}_{\mathrm{VI}}(f(s),s)=0$.

\begin{theorem}\label{thm tau 5 expansion 1}
A series expansion of the $\tau$ function of the fifth Painlev\'e equation at $s=\infty$ is given by
\begin{align}
    \tau_{\mathrm{V}}^{(\infty)} (s,z_2)
    =& \sum_{n \in \mathbb{Z}} e^{2\pi i n  \varrho}\eta^{-2n^2}
    (-1)^{ \frac{1}{2}n(n+1)}
    C_\mathrm{V}(\vec{\theta},\beta+n) \nonumber\\
    &\times
    \braket{ (\eta \theta, \eta^2/4)|
    \Phi_{( \eta \theta, \eta^2/4),\, \left( \eta(\theta - \beta -n) ,\eta^2/4 \right) } ^{\theta_t ^2,*} (s)
    \Phi_{\theta_0^2,0}^{\theta_0^2,*}(z_2)|0 }
    \label{eq tau V s},
\end{align}
where $\vec{\theta}$ stands for $(\theta, \theta_t, \theta_0)$, $\varrho,\, \beta \in \mathbb{C}$, and
\begin{align*}
    C_\mathrm{V}(\vec{\theta},\beta)
    = G(1 + \theta_0 + \theta - \beta ) G(1 - \theta_0 + \theta - \beta )
    G(1+\theta_t + \beta )G(1+\theta_t -  \beta ),
\end{align*}
and the irregular vector $\bra{(\eta(\theta-\beta-n),\eta^2/4)}$ belongs to the completion of the Verma module $M_{\theta_0^2}$. In particular, the function $\tau_{\mathrm{V}}^{(\infty)} (s,z_2)$ satisfies $\tilde{E}_{\mathrm{V}} (f(s),s)=0$.
\end{theorem}

\begin{proof}
We consider the limit of \eqref{eq tau VI s}. By Corollary \ref{cor dual RCB ICB}, the conformal-block part satisfies
\begin{align*}
    &\lim_{\epsilon \to 0}\epsilon^{(\sigma+n)^2 - \theta_1^2 -\theta_0^2 + A(n)}(-1)^{-A(n)}
    \braket{\theta_{\infty}^2|
    \Phi_{\theta_{\infty} ^2,\, (\sigma + n)^2} ^{\theta_{t} ^2,*}(s)
    \Phi_{(\sigma+n)^2,\, \theta_{0}^2}^{\theta_{1}^2,*}(z_1)
    \Phi_{\theta_0^2,\, 0}^{\theta_{0}^2,*}(z_2)
    |0}
    \\
    &=
    \braket{ (\eta \theta, \eta^2/4)|
    \Phi_{( \eta \theta, \eta^2/4),\, \left( \eta(\theta - \beta -n) ,\eta^2/4 \right) } ^{\theta_t ^2,*} (s)
    \Phi_{\theta_0^2,\, 0}^{\theta_{0}^2,*}(z_2)
    |0 },
\end{align*}
where the parameters are scaled as in \eqref{lim 6to5} and
\begin{align*}
    \sigma = \frac{\eta}{2 \epsilon} - \frac{\theta}{2} + \beta, \qquad
    A(n) = - \frac{\eta(\beta+n)}{\epsilon} - (\beta+n)(\theta-\beta-n) - \theta_t^2.
\end{align*}
It therefore remains to analyze the factor
\begin{align*}
    e^{2\pi i n  \varrho} C_{\mathrm{VI}}(\vec{\theta},\sigma+n)
    \epsilon^{-(\sigma+n)^2 + \theta_1^2 +\theta_0^2 - A(n)}(-1)^{A(n)}.
\end{align*}
Using the relations

\begin{align}
    G(1 + x + n) &= (1 + x + n-2) (1 + x + n -3)^2 \cdots (1+x)^{n-1}x^n \Gamma^n(x) G(1+x), \label{eq Barnes relation 1}\\
    G(1+x-n) &= (1+x-n)(1+x-n+1)^2 \cdots (1+x-2)^{n-1} \Gamma^{-n}(x) G(1+x), \label{eq Barnes relation 2}
\end{align}
we obtain
\begin{align*}
    &e^{2\pi i n  \varrho}
    \frac{C_{\mathrm{VI}}(\vec{\theta},\sigma+n)
    \epsilon^{-(\sigma+n)^2 + \theta_1^2 +\theta_0^2 - A(n)}(-1)^{A(n)}}
    {C_{\mathrm{VI}}(\vec{\theta},\sigma)
    \epsilon^{-\sigma^2 + \theta_1^2 +\theta_0^2 - A(0)}(-1)^{A(0)}}\\
    =&\,
    e^{2 \pi i n\varrho'} (-1)^{-\frac{n(n-1)}{2}}\eta^{-2n^2}
    \frac{\prod_{k=\pm 1}G(1 +k \theta_0 + \theta - \beta -n) \prod_{\ell=\pm 1}G(1+ \theta_t +\ell (\beta + n))}
    {\prod_{k=\pm 1}G(1 +k \theta_0 + \theta - \beta) \prod_{\ell=\pm 1}G(1+\theta_t+\ell \beta)},
\end{align*}
where
\begin{align*}
    e^{2 \pi i \varrho'}
    =&\,
    e^{2 \pi i \varrho}
    \frac{\Gamma(1+ \theta_0 + \eta/\epsilon + \beta )\Gamma(1- \theta_0 + \eta/\epsilon + \beta )
    \Gamma(1+\theta_t + (\eta/\epsilon - \theta + \beta))
    \Gamma^2 (1-\eta/\epsilon + \theta - 2\beta)}
    {\Gamma(1+\theta_t - (\eta/\epsilon - \theta + \beta))
    \Gamma^2 (1+\eta/\epsilon - \theta + 2 \beta)}\\
    &\times
    \eta \epsilon^{2(\theta - 2 \beta) -1} (-1)^{-\eta/\epsilon- \theta +2\beta}.
\end{align*}
Thus, after multiplication by a suitable scalar, the function $\tau_{\mathrm{VI}}^{(\infty)} (s,z_1,z_2)$ degenerates to $\tau_{\mathrm{V}}^{(\infty)} (s,z_2)$ in the limit
\begin{equation*}
    \theta_1 = \frac{\eta}{2 \epsilon} + \frac{\theta}{2}, \qquad
    \theta_{\infty} = \frac{\eta}{2 \epsilon} - \frac{\theta}{2}, \qquad
    \sigma = \frac{\eta}{2 \epsilon} - \frac{\theta}{2} + \beta, \qquad
    z_1 = \frac{1}{\epsilon}, \qquad \epsilon\to 0.
\end{equation*}
\end{proof}

A similar degeneration can be obtained by replacing the independent variable by $\epsilon t$ and taking the limit $\epsilon\to 0$ with the parameter scaling \eqref{lim 6to5}. In this way, the $\tau$ function $\tau_{\mathrm{VI}}^{(\infty)}(\epsilon t)$ degenerates to a $\tau$ function of the fifth Painlev\'e equation. We also note that, by taking the limit $z_2\to 0$ with $\eta=1$ in Theorem \ref{thm tau 5 expansion 1}, one recovers the following expansion. Indeed, by the definition of the regular vertex operator,
\[
\lim_{z_2 \to 0}\Phi_{\theta_0^2,0}^{\theta_0^2}(z_2)\ket{0} = \ket{\theta_0^2}.
\]

\begin{theorem}[Conjecture 4.1 in \cite{Nagoya 2015}]\label{thm tau 5 expansion 2}
A series expansion of the $\tau$ function of the fifth Painlev\'e equation at $t=\infty$ is given by
\begin{align*}
    \tau_{\mathrm{V}}^{(\infty)} (t)
    =& \sum_{n \in \mathbb{Z}} e^{2\pi i n  \varrho}
    (-1)^{ \frac{1}{2}n(n+1)}
    C_\mathrm{V}(\vec{\theta},\beta+n)
    \braket{ (\theta, 1/4)|
    \Phi_{( \theta, 1/4),\, \left( \theta - \beta -n ,1/4 \right) } ^{\theta_t ^2,*} (t)
    |\theta_0^2 },
\end{align*}
where $\varrho,\, \beta \in \mathbb{C}$, and
\begin{align*}
    C_\mathrm{V}(\vec{\theta},\beta)
    = G(1 + \theta_0 + \theta - \beta )G(1 -\theta_0 + \theta - \beta )
    G(1+\theta_t +\beta )G(1+\theta_t -\beta ).
\end{align*}
In particular, the function $\tau_{\mathrm{V}}^{(\infty)} (t)$ satisfies $E_{\mathrm{V}} (\sigma_\mathrm{V}(t),t)=0$.
\end{theorem}

This establishes Conjecture 4.1 in \cite{Nagoya 2015} for the expansion of the fifth Painlev\'e $\tau$ function in terms of irregular conformal blocks of type $(0,0,1)$.

\begin{theorem}[Conjecture 4.2 in \cite{Nagoya 2015}] \label{thm tau 4 expansion}
A series expansion of the $\tau$ function of the fourth Painlev\'e equation at $t=\infty$ is given by
\begin{align*}
    \tau_{\mathrm{IV}}^{(\infty)} (t)
    =& \sum_{n \in \mathbb{Z}} e^{2\pi i n  \varrho}
    C_\mathrm{IV}(\vec{\theta},\beta+n)
    \braket{(\theta_*,0,1/4)|
    \Phi_{(\theta_*,0,1/4),(\theta_* - \beta - n,0,1/4)}^{\theta_t ^2,*}(t)|0},
\end{align*}
where $\vec{\theta}$ stands for $(\theta_*, \theta_t)$, $\varrho,\, \beta \in \mathbb{C}$, and
\begin{equation*}
    C_\mathrm{IV}(\vec{\theta},\beta)
    =G(1 + \theta_{*}  - \beta)
    G(1+\theta_t + \beta )G(1+\theta_t -\beta ).
\end{equation*}
In particular, the function $\tau_{\mathrm{IV}}^{(\infty)} (t)$ satisfies $\tilde{E}_{\mathrm{IV}} (f(t),t)=0$.
\end{theorem}

\begin{proof}
We consider the limit of $\tau_{\mathrm{V}}^{(\infty)}(s,z_2)$ in \eqref{eq tau V s}. By Corollary \ref{cor dual limit of ICB 2}, the conformal-block part satisfies
\begin{align*}
    &\lim_{\epsilon \to 0}
    (-1)^{A(n)}
    \epsilon^{-A(n)}
    \exp\left(
        -\frac{\eta(\theta-\beta-n)}{\epsilon}
    \right)
    \braket{ (\eta \theta, \eta^2/4)|
    \Phi_{( \eta \theta, \eta^2/4),(\eta(\theta - \beta -n) ,\eta^2/4 ) } ^{\theta_t ^2,*} (s)
    \Phi_{\theta_0^2,0}^{\theta_0^2,*}(z_2)|0 }\\
    &=\braket{(\theta_*,0,1/4)|
    \Phi_{(\theta_*,0,1/4),(\theta_* - \beta - n,0,1/4)}^{\theta_t ^2,*}(t) |0},
\end{align*}
where the parameters are scaled as in \eqref{lim 5to4} and
\begin{align*}
    A(n)=
    \frac{\beta+n}{\epsilon^2}
    +(\beta+n)\theta_*
    -(\beta+n)^2
    +\theta_t^2.
\end{align*}
It therefore suffices to analyze the remaining factor
\begin{equation*}
    e^{2\pi i n\varrho}
    \eta^{-2n^2}
    (-1)^{\frac{1}{2}n(n+1)}
    C_{\mathrm{V}}(\vec{\theta},\beta+n)
    (-1)^{-A(n)}
    \epsilon^{A(n)}
    \exp\left(
        \frac{\eta(\theta-\beta-n)}{\epsilon}
    \right).
\end{equation*}
Put
\begin{align*}
    e^{2\pi i\varrho'}
    =
    e^{2\pi i\varrho}
    \Gamma(-1/\epsilon^2-\beta)^{-1}
    \epsilon^{1/\epsilon^2+\theta_*-2\beta+1}
    e^{1/\epsilon^2}
    (-1)^{-1/\epsilon^2-\theta_*+2\beta}.
\end{align*}
Then, using  \eqref{eq Barnes relation 1} and \eqref{eq Barnes relation 2}, we obtain
\begin{align*}
   &e^{2\pi i n \varrho}
   \eta^{-2n^2}
   (-1)^{\frac{1}{2}n(n+1)}
   \frac{
   C_{\mathrm{V}}(\vec{\theta},\beta+n)
   (-1)^{-A(n)}
   \epsilon^{A(n)}
   \exp\left(
        \frac{\eta(\theta-\beta-n)}{\epsilon}
   \right)}
   {
   C_{\mathrm{V}}(\vec{\theta},\beta)
   (-1)^{-A(0)}
   \epsilon^{A(0)}
   \exp\left(
        \frac{\eta(\theta-\beta)}{\epsilon}
   \right)}
   \\
   =&
   e^{2\pi i n \varrho'}
   \frac{C_\mathrm{IV}(\vec{\theta},\beta+n)}
   {C_\mathrm{IV}(\vec{\theta},\beta)}.
\end{align*}
Thus, after multiplication by a suitable scalar, the function $\tau_{\mathrm{V}}^{(\infty)} (s,z_2)$ degenerates to $\tau_{\mathrm{IV}}^{(\infty)} (t)$ in the limit
\begin{equation*}
s=t,\quad
    \theta = -\frac{1}{2\epsilon^2} + \frac{\theta_{*}}{2}, \qquad
    \eta = -\frac{1}{\epsilon}, \qquad
    \theta_0 = -\frac{1}{2\epsilon^2} -\frac{\theta_{*}}{2}, \qquad
    z_2 =\frac{1}{\epsilon}, \qquad \epsilon \to 0.
\end{equation*}
\end{proof}

\begin{remark}
By differentiating the second-order differential equations $E_\mathrm{J}(\sigma_\mathrm{J},t)=0$ satisfied by the Hamiltonian functions of the Painlev\'e equations, one obtains bilinear equations for the corresponding $\tau$ functions. For $\mathrm{P_{VI}}$, this bilinear equation corresponds to a bilinear relation for Virasoro conformal blocks with $c=1$ \cite{BS}. It is then natural to ask for an analogue of this bilinear relation for Virasoro conformal blocks with general $c$. In \cite{BST}, such analogues were introduced as bilinear equations for quantum Painlev\'e $\tau$ functions. Since our degeneration scheme is valid for general $c$, the corresponding quantum Painlev\'e $\tau$ functions for $\mathrm{P_V}$ and $\mathrm{P_{IV}}$ can be obtained in the same manner. We plan to return to this point elsewhere.
\end{remark}

\appendix
\section{Appendix}

\begin{lemma}\label{lem asym IVO}
Let $M$ be a non-negative integer. The integral
\begin{equation*}
    \int_{\Delta'} ds\, \phi(s)\prod_{i=1}^n \left[ F_i^M(w,s_i)\,(1-w^r z s_i)^c e^{-X^M(s_i)} \right]
\end{equation*}
admits an asymptotic expansion as $w\to 0$. Its principal part is obtained by term-by-term integration of the Taylor expansion in $w$ over the domain
\begin{equation*}
    \Delta^{(1)}(\gamma_{0,\infty})
    =
    \left\{
    (s_1,\ldots,s_n)\,\middle|\,
    s_n\in\gamma_{0,\infty},\;
    s_k\in\gamma_{0,\infty}(0,s_{k+1})\ (1\le k\le n-1)
    \right\},
\end{equation*}
where $\gamma_{p,q}:[0,1]\to\mathbb{P}^1(\mathbb{C})$ is a smooth path such that
$\gamma_{p,q}(0)=p$ and $\gamma_{p,q}(1)=q$, avoiding the singularities of the integrand; moreover, the path $\gamma_{0,\infty}$ is chosen so that the exponential factor in $\phi(s)$ decays exponentially as $s\to\infty$ along $\gamma_{0,\infty}$. A path $\gamma_{p,q}(p,s)$ is the sub-path of $\gamma_{p,q}$ from $p$ to $s$. More precisely,
\begin{equation}\label{eq appendix asymptotic}
    \int_{\Delta'} ds\, \phi(s)\prod_{i=1}^n \left[ F_i^M(w,s_i)\,(1-w^r z s_i)^c e^{-X^M(s_i)} \right]
    =
    \int_{\Delta^{(1)}(\gamma_{0,\infty})} ds\, \phi(s)\,P_{rM-1}(s,w)
    + O(w^{rM}),
\end{equation}
where $P_{rM-1}(s,w)$ denotes the Taylor polynomial of degree $rM-1$ in $w$ of
\begin{equation*}
    \prod_{i=1}^n \left[F_i^M(w,s_i)\,(1-w^r z s_i)^c \right].
\end{equation*}
\end{lemma}

\begin{proof}
We first remove the factor $e^{-X^M(s_i)}$. Writing
\begin{align*}
    &\int_{\Delta'} ds\, \phi(s)\prod_{i=1}^n \left[F_i^M(w,s_i)(1-w^r z s_i)^c e^{-X^M(s_i)}\right] \\
    =&
    \int_{\Delta'} ds\, \phi(s)\prod_{i=1}^n \left[ F_i^M(w,s_i)(1-w^r z s_i)^c \right] \\
    &+
    \int_{\Delta'} ds\, \phi(s)\prod_{i=1}^n \left[ F_i^M(w,s_i)(1-w^r z s_i)^c\right] \left(\prod_{i=1}^{n}e^{-X^M(s_i)}-1\right),
\end{align*}
we claim that the second term is $O(w^{rM})$. Indeed, if
\[
0<s_1<\cdots < s_i < \frac{1}{2w^r z} < s_{i+1}<\cdots <s_n <\frac{1}{w^r z},
\]
then for $j\le i$ one has $|w^r z s_j|\le 1/2$, hence
\[
|1-e^{-X^M(s_j)}|
\le |X^M(s_j)|
\le K_1 (z s_j)^{M+1} w^{rM},
\]
whereas for $j>i$ one has $|2w^r z s_j|>1$, and therefore
\[
|1-e^{-X^M(s_j)}|
\le 1
< K_1'(z s_j)^M w^{rM}
\]
for some constant $K_1,K_1'>0$.
Thus, the contribution of the second term is $O(w^{rM})$, and it is enough to analyze
\begin{equation}\label{eq appendix reduced integral}
    \int_{\Delta'} ds\, \phi(s)\prod_{i=1}^n \left[ F_i^M(w,s_i)(1-w^r z s_i)^c \right].
\end{equation}

Next, since each $F_i^M(w,s_i)$ is analytic in $w$ on the range of integration, Taylor's theorem gives
\[
F_i^M(w,s_i)=\sum_{m=0}^{rM-1} a_m(s_i)w^m + R_{rM}^{(1)}(w,s_i),
\qquad
|R_{rM}^{(1)}(w,s_i)|\le K_2 |w|^{rM}
\]
for some constant $K_2>0$. It follows that the contribution of the remainder is again $O(w^{rM})$. Consequently, it remains to study
\[
\int_{\Delta'} ds\, \phi(s)\prod_{i=1}^n (1-w^r z s_i)^c.
\]

We now deform the integration domain. Since the integrand is holomorphic in each variable away from $0$, $\infty$, $1/(w^r z)$, and the diagonals $s_i=s_j$, Cauchy's theorem yields
\begin{align*}
    \int_{\Delta'} ds\,\phi(s)\prod_{i=1}^n(1-w^r z s_i)^c
    =&\,
    \int_{\Delta^{(1)}(\gamma_{0,\infty})} ds\,\phi(s)\prod_{i=1}^n(1-w^r z s_i)^c \\
    &-
    \int_{\Delta^{(2)}(\gamma_{1/(w^r z),\infty})} ds\,\phi(s)\prod_{i=1}^n(1-w^r z s_i)^c,
\end{align*}
where
\begin{align*}
    \Delta^{(2)}(\gamma_{p,q})
    :=
    \left\{
    (s_1,\ldots,s_n)\,\middle|\,
    s_n\in\gamma_{p,q},\;
    s_k\in\gamma_{0,s_n}(0,s_{k+1})\ (1\le k\le n-1)
    \right\},
\end{align*}
and $\gamma_{1/(w^r z),\infty}$ is chosen so that the exponential factor in $\phi(s)$ also decays exponentially as $s\to\infty$ along it.

The second term is exponentially small:
\[
\int_{\Delta^{(2)}(\gamma_{1/(w^r z),\infty})} ds\,\phi(s)\prod_{i=1}^n(1-w^r z s_i)^c
=
e^{-\lambda_r\lambda_+/w^r}O(1).
\]
This follows from the change of variables \(s_i=u_i+1/(w^r z)\) together with the standard asymptotics of Kummer's confluent hypergeometric function.

It therefore remains to consider
\begin{align*}
    \int_{\Delta^{(1)}(\gamma_{0,\infty})}ds\,\phi(s)\prod_{j=1}^n(1-w^r z s_j)^c.
\end{align*}
We decompose the domain $\Delta^{(1)}(\gamma_{0,\infty})$ according to the index $i\in\{0,\dots,n\}$ for which
\[
|s_1|<\cdots<|s_i|\le \frac{1}{2|w|^r|z|}
\qquad\text{and}\qquad
\frac{1}{2|w|^r|z|}<|s_{i+1}|<\cdots<|s_n|,
\]
with the obvious convention when $i=0$ or $i=n$. For $j\le i$, one has $|w^r z s_j|\le 1/2$, and hence the Taylor expansion
\[
(1-w^r z s_j)^c=\sum_{m\geq 0} b_m (w^r z s_j)^m
\]
converges uniformly. For $j>i$, the contribution of the remainder is exponentially small, owing to the exponential decay of $\phi(s)$ along $\gamma_{0,\infty}$.

It follows that, after summing over all possible values of $i$, the integral over $\Delta^{(1)}(\gamma_{0,\infty})$ admits an asymptotic expansion obtained by term-by-term integration of the Taylor expansion of
\[
\prod_{j=1}^n(1-w^r z s_j)^c.
\]
This yields the contribution
\[
\int_{\Delta^{(1)}(\gamma_{0,\infty})} ds\, \phi(s)\,P_{rM-1}(s,w)
\]
up to an error of order $O(w^{rM})$. Together with the previous steps, this proves \eqref{eq appendix asymptotic}.

\end{proof}

\end{document}